\documentclass[sigconf, nonacm, pdfa]{acmart}

\usepackage{graphicx}
\usepackage{subcaption}
\usepackage{algorithm}
\usepackage{algorithmic}
\usepackage{tikz}
\usepackage{enumitem}
\usepackage{comment}
\usepackage{xcolor}
\usepackage[normalem]{ulem}
\usepackage{balance}
\usepackage[a-2b]{pdfx}

\usepackage{longtable}
\usepackage{booktabs}   
\usepackage{subcaption} 

\usepackage{svg}
\usepackage{tabularx}
\usepackage{comment}

\usepackage{array}
\usepackage{mathtools}
\usepackage{multirow}
\usepackage{multicol}
\usepackage{hhline}
\usepackage[export]{adjustbox}
\usepackage{makecell}

\usepackage{threeparttable}
\usepackage{url}

\usepackage{todonotes}
\usepackage{enumitem}

\usetikzlibrary{calc}

\usepackage{hhline}
\usepackage{wrapfig}
\usepackage{nicematrix} 
\usepackage[table,xcdraw]{xcolor}

\usepackage{tikz}
\usetikzlibrary{calc}

\newcolumntype{|}{!{\vrule width 0.5pt}}
\newcolumntype{?}{!{\vrule width 1pt}}
\newcolumntype{^}{!{\vrule width 1.2pt}}

\definecolor{myblue}{RGB}{91,155,213}
\definecolor{mydark}{RGB}{0,0,0}

\newcommand{\floor}[1]{\left\lfloor #1 \right\rfloor}

\newcommand{\appropto}{\mathrel{\vcenter{
  \offinterlineskip\halign{\hfil$##$\cr
    \propto\cr\noalign{\kern2pt}\sim\cr\noalign{\kern-2pt}}}}}

\newcolumntype{|}{!{\vrule width 0.5pt}}
\newcolumntype{?}{!{\vrule width 1pt}}
\newcolumntype{^}{!{\vrule width 1.2pt}}

\newcolumntype{|}{!{\vrule width 0.5pt}}
\newcolumntype{?}{!{\vrule width 1pt}}
\newcolumntype{^}{!{\vrule width 1.2pt}}

\definecolor{myblue}{RGB}{91,155,213}
\definecolor{mydark}{RGB}{0,0,0}

\newcommand{\method}{\texttt{RT-RkNN}}
\newcommand{\RkNN}{R$k$NN}
\newcommand{\kNN}{$k$NN}

\newcommand\vldbdoi{10.14778/3819518.3819527}
\newcommand\vldbpages{1963 - 1976}
\newcommand\vldbvolume{19}
\newcommand\vldbissue{9}
\newcommand\vldbyear{2026}
\newcommand\vldbauthors{\authors}
\newcommand\vldbtitle{\shorttitle} 
\newcommand\vldbavailabilityurl{https://github.com/simon2/RT-RkNN}
\newcommand\vldbpagestyle{empty} 

\begin{document}
\title{\method{}: Reverse k Nearest Neighbor Queries as a Graphics Ray Casting Problem}

\author{Zhengyang Bai}
\affiliation{%
  \institution{RIKEN Center for Computational Science}
  \city{Kobe}
  \state{Japan}
}
\email{zhengyang.bai@riken.jp}


\author{Peng Chen}
\authornote{Corresponding author.}
\affiliation{%
  \institution{RIKEN Center for Computational Science}
  \city{Kobe}
  \state{Japan}
}
\email{peng.chen@a.riken.jp}

\author{Mohamed Wahib}
\authornotemark[1]
\affiliation{%
  \institution{RIKEN Center for Computational Science}
  \city{Kobe}
  \state{Japan}
}
\email{mohamed.attia@riken.jp}
\begin{abstract}
Reverse $k$ nearest neighbor (\RkNN{}) queries are fundamental in spatial databases, location-based analytics, and recommendation systems. Existing state-of-the-art techniques rely on spatial pruning supported by R-trees and their variants. However, their pruning effectiveness degrades significantly in challenging scenarios where the number of facilities is small, the user population is dense, or the value of $k$ is large. To overcome these limitations, we formulate the \RkNN{} query in two-dimensional geometric spaces as a graphics ray casting problem, in which users are modeled as rays and facilities are represented as geometric primitives. Based on this formulation, we design the first algorithm and provide an implementation that exploits dedicated hardware ray tracing cores on modern GPUs. This novel approach preserves strong filtering performance even for large values of $k$, dense user populations, and highly sparse facility distributions. Extensive experimental results demonstrate that our method outperforms state-of-the-art algorithms in diverse settings, especially in scenarios where traditional pruning strategies become inefficient.

\end{abstract}

\maketitle

\pagestyle{\vldbpagestyle}
\begingroup\small\noindent\raggedright\textbf{PVLDB Reference Format:}\\
\vldbauthors. \vldbtitle. PVLDB, \vldbvolume(\vldbissue): \vldbpages, \vldbyear.\\
\href{https://doi.org/\vldbdoi}{doi:\vldbdoi}
\endgroup
\begingroup
\renewcommand\thefootnote{}\footnote{\noindent
This work is licensed under the Creative Commons BY-NC-ND 4.0 International License. Visit \url{https://creativecommons.org/licenses/by-nc-nd/4.0/} to view a copy of this license. For any use beyond those covered by this license, obtain permission by emailing \href{mailto:info@vldb.org}{info@vldb.org}. Copyright is held by the owner/author(s). Publication rights licensed to the VLDB Endowment. \\
\raggedright Proceedings of the VLDB Endowment, Vol. \vldbvolume, No. \vldbissue\ %
ISSN 2150-8097. \\
\href{https://doi.org/\vldbdoi}{doi:\vldbdoi} \\
}\addtocounter{footnote}{-1}\endgroup

\ifdefempty{\vldbavailabilityurl}{}{
\vspace{.3cm}
\begingroup\small\noindent\raggedright\textbf{PVLDB Artifact Availability:}\\
The source code, data, and/or other artifacts have been made available at \url{\vldbavailabilityurl}.
\endgroup
}

\section{Introduction}\label{sec:intro}
Reverse $k$ Nearest Neighbor (\RkNN{}) queries are fundamental operators in spatial databases and location-based services, with applications spanning influence analysis, recommendation systems, and spatial impact assessment~\cite{RkNN4trj,RkNN4loc,rknn4Maxinf,rknn,rknn4fac,maxRkNN,maxRkNN2,contimaxrknn}. Given a set of facilities $F$ and users $U$, an \RkNN{} query retrieves all users for whom a given facility is one of their $k$ nearest neighbors. 
Unlike the standard $k$ nearest neighbor (\kNN{}) query that captures local proximity, \RkNN{} characterizes a facility's \emph{influence set} or effective market reach, supporting operational efficiency and protecting substantial capital investments in domains such as retail site selection and delivery services~\cite{milliondollar}. 
Beyond spatial applications, \RkNN{} has gained renewed attention in artificial intelligence, where modeling influence often outweighs similarity alone. It refines density estimation and cluster expansion in DBSCAN variants~\cite{RNN-DBSCAN,KR-DBSCAN}, and mitigates hubness in graph neural networks, promoting more balanced graph structures and improved representation learning~\cite{hubness}.
Due to their utility, significant research has focused on developing efficient \RkNN{} processing techniques~\cite{rknn,six,yl01,finch,InfZone,tpl,slice,vor-tree,csd-rknn,tplpp,vol2,safar2009voronoi}, as well as variants like dynamic, continuous, probabilistic, and batched \RkNN{}~\cite{dynamicRkNN,dynamic2,continousRkNN,BRkNN,InfZone2,lazy-update,continuous2,prob-rknn,prob-rknn2,moving-obj,4511444,pathRknn,rknnmob}.

A critical distinction from the standard \kNN{} query is the nature of their search spaces. While a \kNN{} query operates within a compact, query-centered region, an \RkNN{} query must examine user-centric influence regions distributed across the entire data space, resulting in a significantly larger and more complex search area. To manage this complexity, most state-of-the-art solutions rely on spatial pruning strategies to eliminate candidates early. These include region-based pruning~\cite{six,slice} and half-space pruning~\cite{tpl,InfZone,tplpp,finch}, typically built upon spatial indexes like the R-tree~\cite{R-tree} or its variants (e.g., R*-tree~\cite{R*-tree}). These methods leverage geometric relationships to achieve strong practical performance on large datasets.
Despite their success, the efficiency of these pruning-based strategies is predicated on the selectivity of geometric relationships. This reliance becomes a critical limitation in realistic and increasingly common scenarios where this selectivity diminishes. 
For instance, California's 337 hospitals serve millions of patients annually~\cite{hospital}, while around 50 last-mile facilities deliver 2.3 million daily packages to 8.8 million residents in New York City~\cite{nyc_comptroller_2025}. In such cases, the expansive influence region of each facility severely undermines spatial pruning, motivating the pursuit of a robust \RkNN{} method that does not degrade under these challenging conditions.

We identify three specific scenarios where traditional pruning-based \RkNN{} algorithms face significant performance degradation: 
1) \emph{Small Facility Sets:} A small number of facilities leads to expanded potential influence regions for each, drastically reducing the number of candidates that can be pruned. 
2) \emph{Large User Populations:} As the cardinality of $U$ grows, spatial pruning retains a much larger pool of candidate users, leading to a prohibitive increase in verification costs.
3) \emph{Large $k$ Values:} Increasing $k$ weakens the geometric criteria for discarding candidates, as the conditions for exclusion become harder to satisfy, thereby diminishing the pruning effect.

To address these challenges, we propose a fundamental formulation of the two-dimensional (2D) \RkNN{} problem, called~\method{}. We model it as a three-dimensional (3D) graphics ray casting problem, where users are treated as rays and facilities are encoded as geometric primitives. This novel formulation allows the \RkNN{} evaluation to be mapped to massively parallel ray-primitive intersection tests. Crucially, the core geometric operations in this model align perfectly with the capabilities of GPU Ray Tracing (RT) cores, which were originally designed for 3D rendering. By leveraging these specialized hardware accelerators, our approach minimizes warp divergence, exploits massive parallelism, and delivers scalable \RkNN{} performance that remains stable even in settings where traditional spatial pruning becomes ineffective.

This paper makes the following contributions:
\begin{itemize}[leftmargin=2.5mm]
    \item We formulate 2D \RkNN{} queries as a 3D graphics ray casting problem and establish the formal equivalence between ray--primitive occlusion and \RkNN{}'s spatial pruning relationships.
    \item We provide the first algorithm derived from this new formulation, which can leverage mature technologies in computer graphics and enables massively parallel execution of ray--primitive intersection tests.
    \item We present the application of GPU ray tracing cores to \RkNN{} query processing, demonstrating how hardware-accelerated ray--primitive intersection tests can be effectively leveraged for efficient spatial pruning and user verification. 
    \item Extensive experiments on real-world datasets with up to more than 23 million users demonstrate that our approach outperforms state-of-the-art baselines across diverse parameter settings. Notably, our method performs particularly well in scenarios where traditional spatial pruning becomes inefficient, such as with large $k$ values, high user cardinality, or sparse facility distributions.
\end{itemize}


\section{Background \& Related work}\label{sec:pre}
In this section, we first provide the formal definition of the target problem and address the scope of this paper in \autoref{sec:def}.  We review the \RkNN{} related work in \autoref{sec:related_work}, and the use of RT cores in spatial queries in \autoref{related-rt}.

\subsection{Problem Statement}\label{sec:def}
\RkNN{} queries are classified into \textit{monochromatic \RkNN{} queries} and \textit{bichromatic \RkNN{} queries}~\cite{rknn}.

\paragraph{\textbf{Monochromatic \RkNN{}}}
Let $P=\{p_1,\ldots,p_n\}$ be a set of points in a metric space. For any point $p\in P$, let $\mathrm{NN}_k(p;P)$ denote the set of the $k$ nearest neighbors of $p$ within $P\setminus\{p\}$.  
A monochromatic \RkNN{} query for a point $q\in P$ returns
\[
\mathrm{RkNN}(q;P)
    = \{\, p\in P\setminus\{q\} \mid q \in \mathrm{NN}_k(p;P) \,\}.
\]
In other words, $p$ is an \RkNN{} of $q$ if $q$ is among the $k$ closest points to $p$ in the same dataset.

\paragraph{\textbf{Bichromatic \RkNN{}}}
Let $F=\{f_1,\ldots,f_m\}$ be a set of facilities and $U=\{u_1,\ldots,u_n\}$ be a set of users, where $F$ and $U$ lie in the same metric space and $F\cap U=\emptyset$.
For any user $u\in U$, let $\mathrm{NN}_k(u;F)$ denote its $k$ nearest facilities in $F$.  
A bichromatic \RkNN{} query for a facility $f\in F$ returns
\[
\mathrm{RkNN}(f;F,U)
    = \{\, u\in U \mid f \in \mathrm{NN}_k(u;F) \,\}.
\]
That is, $u$ is an \RkNN{} of $f$ if $f$ is one of the $k$ nearest facilities to $u$.

\RkNN{} queries are challenging because the membership condition for a user $u$ depends on how the distance $\mathrm{dist}(u,f)$ compares to its distances to \emph{all} other facilities in $F$. Formally, $u$ is an \RkNN{} of $f$ if and only if
\[
\bigl|\{\, a \in F \setminus \{f\} \mid \mathrm{dist}(u,a) < \mathrm{dist}(u,f) \,\}\bigr| < k .
\]
This condition requires determining the rank of $\mathrm{dist}(u,f)$ within the entire distance set $\{\mathrm{dist}(u,a)\}_{a \in F}$. Unlike forward \kNN{}, which only identifies the $k$ smallest distances for each query point, the reverse condition depends on the \emph{global ordering} of distances around every user. Consequently, \RkNN{} queries inherently involve many distance comparisons per user, making them mathematically global and significantly more complex.

\textbf{Scope of this paper}. 
Following existing \RkNN{} techniques that primarily target 2D spatial data that arise in location-based services~\cite{six,finch,InfZone,slice}, we also put the focus of this paper on 2D location data.
Although our approach applies to both monochromatic and bichromatic \RkNN{} queries, we focus on the bichromatic setting in our performance evaluation for easier presentation as it represents the more general case with the monochromatic case reducible to it. The performance of our approach on monochromatic \RkNN{} queries is only briefly discussed.
Although, as mentioned above, \RkNN{} queries have been extensively studied with various variants, we focus on the most widely used static queries in Euclidean space in this paper.
Finally, given the substantially larger memory in modern machines, even our largest dataset (containing over 23 million points) occupies only about 218 MB, far below the capacity of a typical PC.
We therefore assume that all data can be fully loaded into main memory. 
As a result, we do not evaluate disk I/O performance in this paper and instead report runtime directly.

\begin{figure*}[t]
    \centering
    \subfloat[SIX filtering example.]{
        \includegraphics[width=0.215\linewidth]{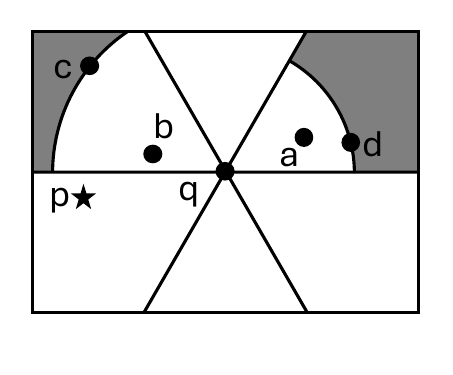}
        \label{fig:six}
    }
    \hfill 
    \subfloat[TPL filtering example.]{
        \includegraphics[width=0.22\linewidth]{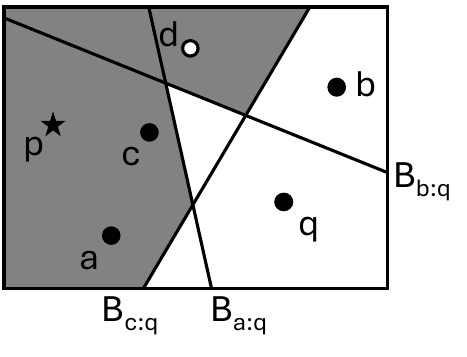}
        \label{fig:tpl}
    }
    \hfill 
    \subfloat[InfZone pruning example.]{%
        \includegraphics[width=0.22\linewidth]{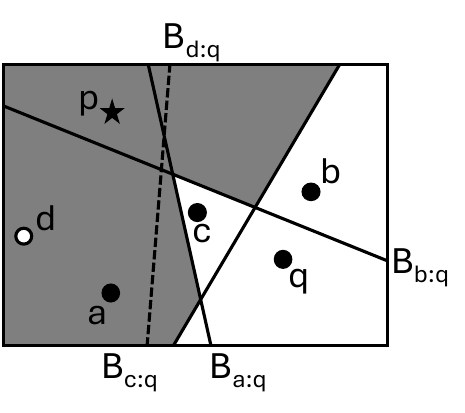}%
        \label{fig:inf}
    }
    \hfill 
    \subfloat[SLICE pruning example.]{%
        \includegraphics[width=0.23\linewidth]{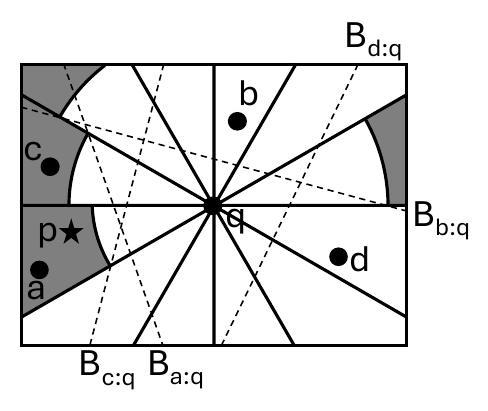}%
        \label{fig:slice}
    }
    \caption{Illustration of SIX~\cite{six} filtering, TPL~\cite{tpl} filtering, InfZone~\cite{InfZone} pruning, and SLICE~\cite{slice} pruning  ($k=2$).}
    \label{fig:system_matrix_introduction}
\end{figure*}

\subsection{Existing \RkNN{} Algorithms}\label{sec:related_work}
While early studies on reverse nearest neighbor (RNN) queries relied on substantial preprocessing to enable efficient query evaluation~\cite{rknn,yl01,vor-tree,vol2,safar2009voronoi,csd-rknn}, such preprocessing introduces significant overhead and lacks the flexibility required for dynamic updates in real-time services. We therefore focus on four representative preprocessing-free techniques, selected for their wide adoption and strong performance.

\textbf{SIX}~\cite{six}, known as the first algorithm that does not need any pre-computation, uses a \textit{regions-based} pruning strategy which partitions the space around the query point into six equal angular partitions of $60^\circ$ each.
For a given partition $P$, the distance from the query to its $k$-th nearest facility in that region defines a filtering threshold: any user located in $P$ at a distance greater than the threshold cannot be an \RkNN{} of the query facility. 
\autoref{fig:six} illustrates the idea that users lying in the shaded area are pruned because at least $k$ facilities within the same partition are closer to them than the query.
The remaining users are treated as candidates to be verified in the subsequent verification stage, which is performed by issuing a range query centered at the candidate with radius $dist(u,q)$, where $u$ is the candidate user and $q$ is the query facility.
The candidate is confirmed as an \RkNN{} result if and only if the range query contains no more than $k$ facilities.
However, since a separate range query is required for each candidate user, the verification phase leads to substantial overhead.

\textbf{TPL}~\cite{tpl} is the first algorithm which applied spatial pruning based on \textit{half-space} pruning to \RkNN{} queries and inspired many follow up works~\cite{finch,InfZone,tplpp}, including ours. Given a query point $q$ and a facility $f$, the perpendicular bisector $B_{f:q}$ divides the space into two half-spaces: the region $H_{f:q}$ that contains $f$ and the opposite region $H_{q:f}$. 
Any point $p$ lying in region $H_{f:q}$ indicates that $p$ is closer to facility $f$ rather than $q$. A point lying in at least $k$ such half-spaces cannot regard $q$ as one of its $k$ nearest facilities and can therefore be pruned. 
The TPL algorithm iteratively accesses unpruned facilities in increasing order of their distance to $q$, and each accessed facility contributes a half-space that potentially enlarges the prunable region.
Consider the example in ~\autoref{fig:tpl}, where the bisectors of facilities $a$, $b$, and $c$ partition the space. When $k=2$, the shaded region corresponds to points that fall in at least two of the half-spaces $H_{a:q}$, $H_{b:q}$, and $H_{c:q}$, and thus any point in this region can be
safely filtered. Facility $d$, however, lies in the region that can be filtered, so it is not used to build bisectors for spatial pruning. 
User $p$ lies inside the prunable region and therefore cannot be an \RkNN{} candidate. 
As TPL’s filtering phase may leave false positives like SIX, the algorithm includes a subsequent refinement stage to verify the remaining candidates.

\textbf{Influence Zone (InfZone)}~\cite{InfZone} relies on the concept of the \textit{influence zone}, which represents the region in which a user is guaranteed to be an \RkNN{} of the query point $q$. 
In other words, a point $p$ is an \RkNN{} of $q$ if and only if $p$ lies inside the influence zone. 
To estimate the influence zone, InfZone incrementally intersects the half-spaces generated by the facilities. 
The construction of the influence zone, illustrated in \autoref{fig:inf}, is also driven by the half-spaces of facilities relative to the query point \(q\). As each facility is processed, the region is filtered by removing any area excluded by at least \(k\) facilities, and the remaining area is the influence zone. Facilities whose bisectors cannot affect the influence zone (like facility \(d\) in \autoref{fig:inf}) are ignored.
The InfZone algorithm builds the zone by tracking the intersections of unpruned bisectors. A critical insight is that the influence zone is star-shaped~\cite{inf_star}, which enables an efficient pruning strategy: a facility \(f\) is unnecessary if \(\min dist(f, v) > dist(v, q)\) for every convex vertex \(v\) of the current zone. However, a direct check against all \(O(m^2)\) vertices, where $m$ is the number of bisectors, is computationally heavy. To mitigate this, two inexpensive filters are applied first: (1) a facility \(f\) can be directly pruned if
\begin{equation}\label{eq:cheapinf}
    dist(f, q) > 2 \times \max_{v \in V} dist(v, q),
\end{equation}
and (2) a facility \(f\) cannot be pruned if
\begin{equation}
    dist(f, q) < 2 \times \min_{p \in E} dist(p, q),
\end{equation}
where \(V\) is the vertex set and \(E\) is the boundary of the influence zone. The expensive vertex check is only performed for facilities that do not meet either of these two conditions, which significantly reduces the overall computation cost.
The verification stage in InfZone consists of a single geometric containment check. Any user located inside the influence zone is directly reported as an \RkNN{}, whereas any user outside it is immediately filtered out. This mechanism allows the algorithm to avoid costly distance computations and eliminates the need for a candidate examination process.

\textbf{SLICE}~\cite{slice} enhances the filtering power of the regions-based pruning, which was first proposed in the SIX approach~\cite{six} while retaining its computational efficiency. 
SLICE partitions the space around the query point $q$ into 12 regions of equal angular extent, a number determined to be optimal. Within a partition \(P\), a facility \(f\) defines two key arcs: an upper arc, which describes the area that can be pruned by \(f\), and a lower arc, which defines the area that is safe from \(f\). 
For any partition $P$ and facility $f$, the radius of two arcs can be determined in $O(1)$ time by evaluating the maximum and minimum subtended angles of $f$ over $P$, which are equivalent to the distances from the intersection points between $f$’s bisector $B_{f:q}$ and the two radial boundaries of $P$.
During the pruning phase, for each partition, SLICE maintains the \(k\)-th smallest upper arc as its \textit{bounding arc} (\(r^R_P\)). 
Any point \(p\) in \(P\) with \(dist(p, q) > r^R_P\) is guaranteed to be pruned by at least \(k\) facilities and is filtered. 
As shown in \autoref{fig:slice}, the shaded area is the space pruned by facilities $a, b, c$ and $d$.
We can see that in SLICE, each facility can potentially contribute to the spatial pruning for multiple partitions.
Similar to SIX and TPL, SLICE's filtering phase may also leave false positives; therefore, a verification phase is necessary.
For the verification phase, candidate users are those lying inside their partition's bounding arc. A facility is considered \textit{significant} for a partition if its lower arc is smaller than the partition's bounding arc, meaning it can potentially prune candidates. SLICE creates a sorted \textit{significant list} of these facilities for each partition. To verify a candidate user \(u\), the algorithm checks facilities in this sorted list, counting how many prune \(u\). The process stops early if \(k\) pruning facilities are found, which means $u$ can be filtered or if the lower arc of the current facility exceeds \(dist(u, q)\), which means $u$ is confirmed as an \RkNN{} result of $q$.

\begin{figure*}[t]
    \centering
    \includegraphics[width=1\linewidth]{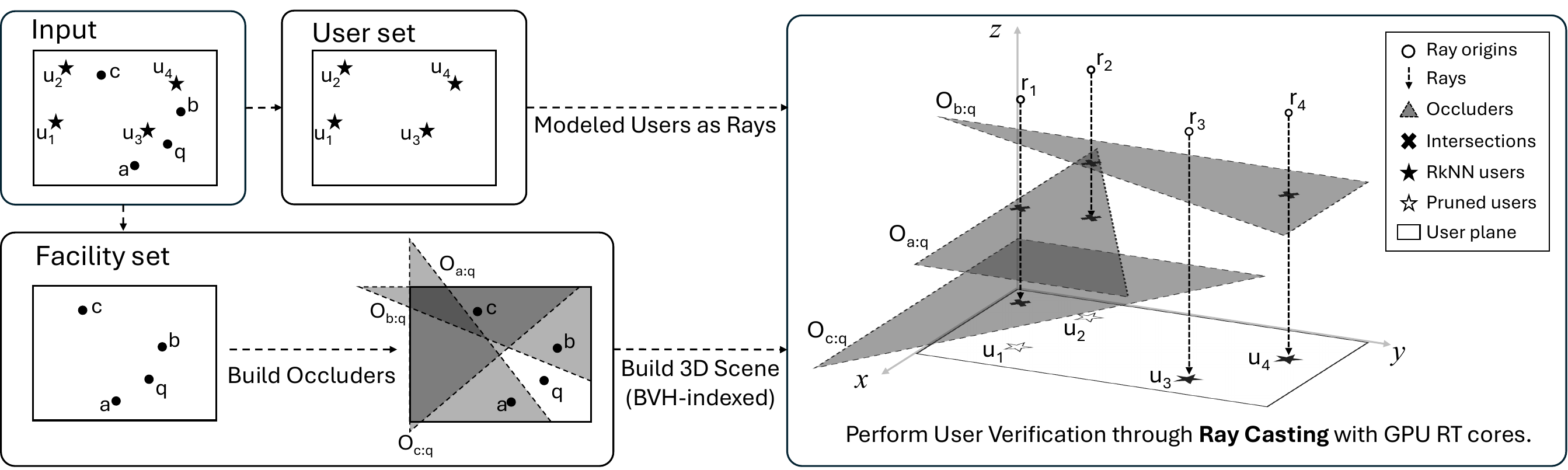}
    \caption{Overview of the \method{} workflow. The input is separated into a user set and a facility set. Users are modeled as rays, while facilities are used to construct occluders based on the invalid sides of pairwise bisectors. All occluders are embedded as layered triangles in a 3D scene indexed by a BVH. During query processing, each user ray is cast perpendicularly to $x,y$-plane. A user is reported as an \RkNN{} result if and only if its ray intersects fewer than $k$ occluders.}
    \Description{overview.}
     \label{fig:overview}
\end{figure*}

\subsection{RT Cores in Spatial Queries}\label{related-rt}
In recent years, 
RT cores have transcended graphics, emerging as a powerful primitive for accelerating database operations including indexing, scan, and query~\cite{RTIndex, RTScan, RayDB}, as well as spatial and distance-based queries such as \kNN{}~\cite{RTNN,RTkNNS,Arkade}, high-dimensional approximate nearest neighbor search~\cite{JUNO}, DBSCAN clustering~\cite{RTDBSCAN}, spatial joins~\cite{RayJoin}, and spatial indexing~\cite{LibRTS}. By encoding spatial proximity as geometric primitives and traversing a Bounding Volume Hierarchy (BVH)~\cite{bvh} on dedicated hardware, RT cores sidestep the branching that degrades standard SIMT execution, reducing distance comparisons to hardware-level intersection tests and enabling concurrent evaluation of millions of user--facility relationships at low latency.
Among these tasks, RT-core-accelerated \kNN{} is the most closely related to our problem, as both involve point-to-point distance evaluation. 
However, all existing \kNN{} methods~\cite{Arkade,RTNN,RTkNNS} using RT cores adopt a query-centered formulation: each point acts as a ray origin, and candidate neighbors are identified through intersection tests or distance filtering relative to that single origin of query point. While this design aligns well with the objective of \kNN{} search, it is fundamentally incompatible with the semantics of \RkNN{}, where the goal is to determine for which users a facility remains competitive against all other facilities. Consequently, query-centered formulations cannot capture the required cross-facility ranking relationships inherent in \RkNN{}, thereby motivating the development of a distinct geometric formulation.

\section{Proposed Approach: \method{}}\label{sec:rt-rknn}
In this section, we describe how an \RkNN{} query can be formulated as a graphics ray casting problem, prove its correctness, and present our algorithm based on this new geometric interpretation that directly aligns with dedicated hardware ray tracing cores on modern GPUs.
{\color{black}To offer a high-level view of the proposed method, \autoref{fig:overview} summarizes the end-to-end \method{} workflow, from inputs to constructing occluders and the BVH-indexed scene, and finally performing ray-casting–based user verification on GPU RT cores.}

\subsection{\RkNN{} Formulation as Graphics Ray Casting}\label{sec:reformulation}
Ray casting, a fundamental operation in modern computer graphics~\cite{raycasting}, determines how rays interact with geometric primitives in a scene. A ray is emitted from an origin and travels through a 3D environment while reporting intersections with objects such as triangles. 
Formally, in ray casting, a ray is defined as
\begin{equation}
\mathbf{r}(t) = \mathbf{o} + t \cdot \mathbf{d}, \label{eq:ray}
\end{equation}
where \( \mathbf{o} \in \mathbb{R}^3 \) is the ray origin, \( \mathbf{d} \in \mathbb{R}^3 \) is the ray direction, and \( t \in [t_{\min}, t_{\max}] \) specifies the ray's active interval. 
A ray intersects a geometric primitive $P$ if and only if
\begin{equation}
\{\mathbf{r}(t) \mid t \in [t_{\min}, t_{\max}]\} \cap P \neq \emptyset .
\label{eq:ray-intersects-P}
\end{equation}
When \(P\) is a triangle with vertices \( \mathbf{v}_0, \mathbf{v}_1, \mathbf{v}_2 \in \mathbb{R}^3 \), the widely used Möller--Trumbore test~\cite{moller} determines an intersection by solving
\begin{equation}
\mathbf{r}(t) = (1 - a - b)\mathbf{v}_0 + a\mathbf{v}_1 + b\mathbf{v}_2, \label{eq:intersection}
\end{equation}
for barycentric coordinates \( a \ge 0 \), \( b \ge 0 \), and \( a + b \le 1 \). These concepts, such as rays, triangle primitives, and ray--primitive intersection tests, form the basis for the geometric constructions used in the following subsection.
Although a single ray may intersect multiple primitives, real-time graphics pipelines typically limit the number of processed intersections for performance reasons~\cite{ue}. 
This operational model naturally aligns with the structure of \RkNN{} queries.

Half-space pruning techniques for \RkNN{} rely on perpendicular bisectors between facilities, where each bisector separates space into a region in which the query facility is closer and a region in which a competing facility is closer. We reinterpret this geometric relationship directly as a graphics ray casting problem: users act as ray origins, and the invalid half-spaces induced by facility pairs are encoded as occluding geometric primitives. Under this formulation, determining whether a user belongs to the \RkNN{} result reduces to counting the number of occluder primitives intersected by their corresponding ray. If a ray intersects at least $k$ such occluders, then at least $k$ competing facilities are closer to that user than the query facility, and the user can be pruned.

For each facility pair with respect to a query facility, the invalid side of their bisector can be represented as one or more triangular occluders. 
\autoref{fig:occluder-cases} illustrates the four possible occluder construction scenarios:
(a) The normal case, where the invalid side forms a triangle naturally;
(b) The extended case, where the invalid region of a non-vertical and non-horizontal bisector does not naturally form a triangular shape because the bisector extends beyond the bounded domain, and thus the exact polygonal boundary is replaced by a single triangle that fully covers the intended invalid space;
(c) The vertical bisector case; and
(d) The horizontal bisector case.
When the bisector is neither vertical nor horizontal, the invalid region forms a triangle. 
Otherwise, the invalid region becomes rectangular and can be modeled using two triangles.
Each occluder is constructed by intersecting the bisector with the rectangular domain boundary, producing two intersection points. Together with one or two boundary vertices lying on the invalid side, these points form either a single triangle or two triangles, depending on the shape of the invalid region. 

\begin{figure}[t]
    \centering

    \begin{subfigure}{0.43\linewidth}
        \centering
        \includegraphics[width=\linewidth]{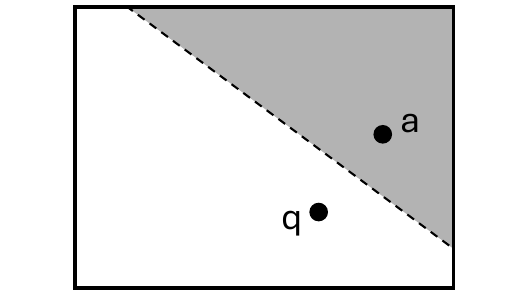}
        \caption{Normal case.}
        \label{fig:a}
    \end{subfigure}
    \hfill
    \begin{subfigure}{0.45\linewidth}
        \centering
        \includegraphics[width=\linewidth]{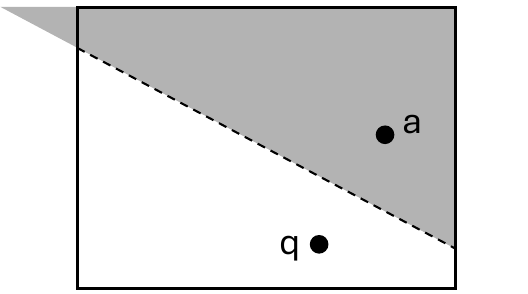}
        \caption{Extended case.}
        \label{fig:b}
    \end{subfigure}

    \vspace{0.3cm}

    \begin{subfigure}{0.45\linewidth}
        \centering
        \includegraphics[width=\linewidth]{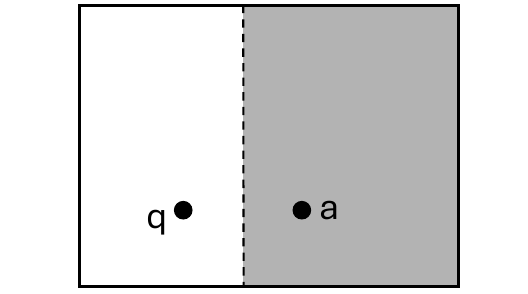}
        \caption{Vertical case.}
        \label{fig:c}
    \end{subfigure}
    \hfill
    \begin{subfigure}{0.45\linewidth}
        \centering
        \includegraphics[width=\linewidth]{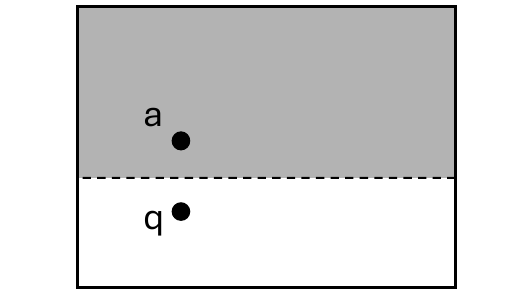}
        \caption{Horizontal case.}
        \label{fig:d}
    \end{subfigure}

    \caption{Four occluder construction scenarios.}
    \label{fig:occluder-cases}
\end{figure}

As illustrated in {\color{black}the right side of} \autoref{fig:overview}, we embed all the occluders of corresponding facilities into a simple 3D \emph{scene} to align the representation with the ray casting direction and to avoid geometric degeneracies. 
Users lie on a plane (e.g., $z = 0$), and rays are cast perpendicularly through the scene from the same $(x,y)$ coordinates as the users but with a larger $z$ value. The occluders are positioned as horizontal layers parallel to the $xy$-plane, each placed at a distinct height, ensuring rays traverse the occluders in a consistent front-to-back order before reaching the user plane.
In this 3D embedding, each user corresponds to a single ray shot through the layered occluders. 
The \RkNN{} query reduces to counting the number of occluders intersected by this ray. 
If the ray hits at most $k\!-\!1$ occluders, the user does not have more than $k$ closer facilities than query $q$ and remains as one of the query results. 
Conversely, if it intersects $k$ or more occluders, the user can be pruned. 
Because the occluders are arranged as discrete layers, intersections occur in strictly ordered depth, enabling early termination once the hit count reaches~$k$.
\autoref{fig:overview} illustrates the example with $k=2$: rays $r_3$ and $r_4$ intersect fewer than two occluders and are reported as \RkNN{} results of $q$, while $r_1$ and $r_2$ each hit at least two occluders and the corresponding users are pruned.


To formalize the notion of \textit{occluder}, \textit{scene} and \textit{ray}, we provide the formal notation used in our geometric representation.
\begin{definition}[Occluder]
\label{def:triangle}
Let $R \subset \mathbb{R}^2$ be a rectangular domain, and let $a, q \in R$ be
two facilities with coordinates $(x_a,y_a)$ and $(x_q,y_q)$, where $q$ is the query facility and $B_{a:q}$ denotes the
perpendicular bisector between $a$ and $q$. 
If $B_{a:q}$ is neither vertical nor horizontal, among the four vertices of $R$
we select the vertex lying on the invalid side of $B_{a:q}$ and farthest in the
direction from $a$ toward $q$, denoted by $v$. The bisector intersects the two
boundary edges incident to $v$, producing two distinct intersection points
$p_1$ and $p_2$.
If $B_{a:q}$ is vertical or horizontal, there are two farthest vertices of $R$
on the invalid side of the bisector with equal distance to $B_{a:q}$, denoted by
$v_1$ and $v_2$.
The bisector intersects the two boundary edges of space $R$ are denoted as $p_1$ and $p_2$.
To embed the geometry into 3D space, we assign a unique height
$z_{a:q} > 0$ shared by all lifted points
($v, p_1, p_2$ or $v_1, v_2, p_1, p_2$) such that no two occluders are associated
with different facility pairs sharing the same $z$-coordinate.

The occluder corresponding to the facility pair $(a,q)$ is then defined as
\begin{equation}\small
O_{a:q} =
\begin{cases}
\triangle(v, p_1, p_2), 
& \text{if } x_a \neq x_q \text{ and } y_a \neq y_q,\\[6pt]
\triangle(v_1, p_1, p_2) \cup \triangle(v_1, v_2, p_2), 
& \text{if } x_a = x_q \text{ or } y_a = y_q.
\end{cases}
\end{equation}
\end{definition}

\begin{definition}[Scene]
\label{def:scene}
Let $F$ be the set of all facilities and let $q \in F$ be the query facility.
For every facility $a \in F \setminus \{q\}$, let $O_{a:q}$ denote the occluder associated with the pair $(a,q)$ as defined in Definition~\ref{def:triangle}. The \emph{scene} corresponding to query facility $q$ is defined as the set:
\begin{equation}
\mathcal{T}_q = \bigl\{\, O_{a:q} \;\bigm|\; a \in F \setminus \{q\} \,\bigr\}.
\end{equation}
\end{definition}

\begin{definition}[Ray]
For a user $u$ with its coordinates $\mathbf{u} \in \mathbb{R}^2$, as rays are cast from the same $(x,y)$-coordinates of users and projected perpendicular to $xy$-plane, we define the ray corresponding to $u$ as:
\begin{equation}
\mathbf{r}_{u} = (\mathbf{u},z_u) + t \cdot (0,0,-1), \label{eq:ray-u}
\end{equation}
based on \eqref{eq:ray}, where $z_u$ is large enough to ensure that the ray origin is higher than all occluders $O \in \mathcal{T}_q$ and $t \in [0,z_u]$.
\label{def:ray}
\end{definition}

\subsection{Theoretical Equivalence and Correctness}
To demonstrate that the ray casting formulation faithfully preserves the semantics of \RkNN{}, we define the hit count $\mathcal{H}_q(\mathbf{r}_u)$ as the number of occluders intersected by the ray $\mathbf{r}_{u}$:
\begin{equation}
    \mathcal{H}_q(\mathbf{r}_u)
    =
    \left|
      \left\{
        O \in \mathcal{T}_q
        \,\middle|\,
        \exists t \in [0,z_u]
        \text{ satisfying }
        \eqref{eq:ray-intersects-P}
      \right\}
    \right|.
    \label{eq:H-definition}
\end{equation}
Using this definition, our ray casting interpretation becomes equivalent to the classical \RkNN{} condition:
\begin{equation}
    u \in \mathrm{RkNN}(q)
    \quad\Longleftrightarrow\quad
    \mathcal{H}_q(\mathbf{r}_u) < k .
    \label{eq:visibility_equivalence}
\end{equation}

\begin{figure}[t] 
  \centering \includegraphics[width=0.9\linewidth]{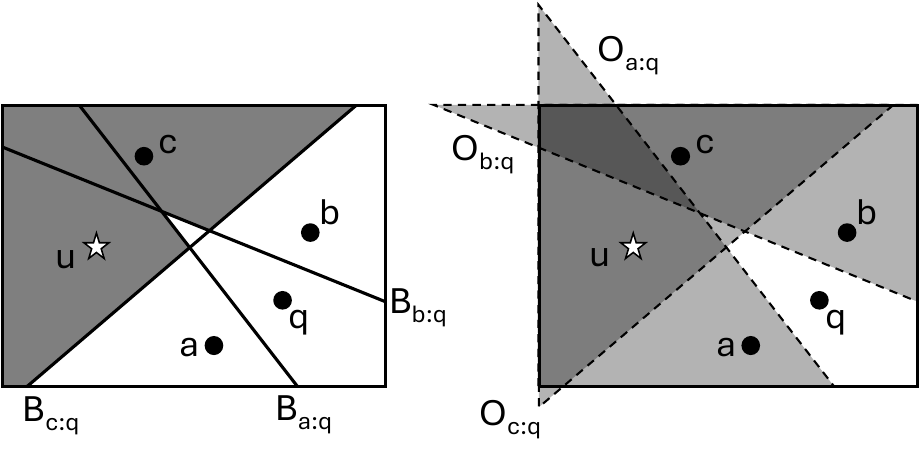} 
  \caption{Equivalence between classical \RkNN{} with ray-casting-based formulation($k=2$) . Half-space pruning (left). Ray casting formulation (right).} 
  \Description{3d-rknn illustration.} 
  \label{fig:rt=rknn} 
\end{figure}

\begin{lemma}[Correctness of Ray Casting Formulation]
\label{lem:visibility-correctness}
For any facility $q$ and user $u$, let $\mathbf{r}_u$ be defined as in Definition~\eqref{eq:ray-u} and $\mathcal{H}_q(u)$ be defined as in Equation~\eqref{eq:H-definition}. 
Then $u$ is a reverse $k$ nearest neighbor of $q$ if and only if $\mathcal{H}_q(\mathbf{r}_u) < k$, i.e., Equation~\eqref{eq:visibility_equivalence} holds.
\end{lemma}

\begin{proof}
Consider the perpendicular bisector between $q$ and any competing facility $f$.
Consistent with concepts used in related work, the invalid side of this bisector is precisely the region in which $f$ is closer to $u$ than $q$. 
Whenever a ray $\mathbf{r}_u$ cast from $u$ perpendicularly to the user plane intersects the occlusion triangle, it means $u$ lies in a space indicating that $f$ is strictly closer to $u$ than $q$. 
Each intersection counted in $\mathcal{H}_q(\mathbf{r}_u)$ therefore corresponds to one distinct facility is closer than $q$ with respect to $u$.

Therefore, if $\mathcal{H}_q(\mathbf{r}_u) \ge k$, then at least $k$ facilities are closer to $u$ than $q$, so $q$ cannot be among the $k$ nearest neighbors of $u$, implying $u \notin \mathrm{RkNN}(q)$. 
Conversely, if $\mathcal{H}_q(u) < k$, then fewer than $k$ facilities closer than $q$ for $u$, meaning that $q$ is within the $k$ nearest facilities to $u$, and therefore $u \in \mathrm{RkNN}(q)$. This completes the proof.
\end{proof}

\autoref{fig:rt=rknn} shows the equivalence between classical \RkNN{}'s half-space pruning and our ray casting formulation using an example of finding \RkNN{} of facility $q$ with competitor facilities $a,b$, and $c$. 
The shaded area in the figure on the left demonstrates the pruned space under two invalid sides of bisectors, which corresponds to the regions shaded with two or more gray levels in figure on the right, which demonstrates an area covered by at least two occluders in a top-down view. Users like $u$ that lie in these areas can be filtered directly.

\subsection{Data Indexing and Proposed Algorithm} \label{alg}
However, in our ray casting formulation, solving Möller–Trumbore intersection tests in Equation~(\ref{eq:intersection}) for all ray–primitive pairs are impractical as the number of rays and primitives increases.
To address this, we apply a well-established technique in computer graphics, bounding volume hierarchy (BVH)~\cite{bvh}, to index the primitives, which significantly reduces the number of intersection tests.
A BVH is a tree data structure that organizes geometric primitives into a hierarchy of nested bounding volumes, typically axis-aligned bounding boxes (AABBs). Each node in the tree represents an AABB that contains either a group of child AABBs or primitives. This structure allows the ray--primitive intersection test to be performed as an efficient single-path ray traversal over the BVH. 

\begin{figure}[t]
  \centering
  \includegraphics[width=1\linewidth]{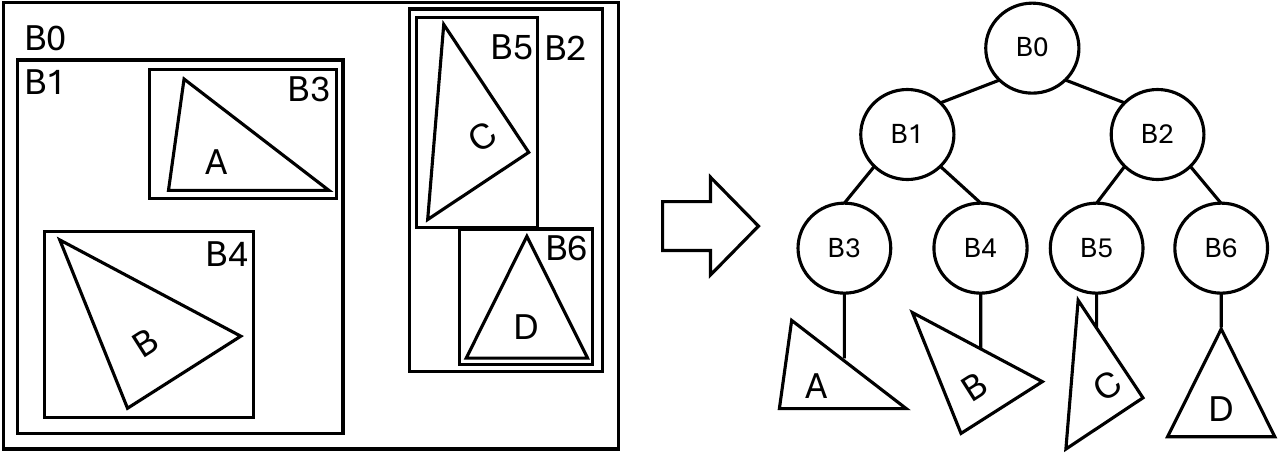}
  \caption{A BVH example (right) corresponding to a 2D scene (left) with four distinct triangle primitives (A, B, C, and D). The tree nodes, labeled B0 to B6, represent the corresponding AABBs in the 2D space. Primitives always appear as leaf nodes in the BVH.}
  \label{fig:bvh}
\end{figure}

\autoref{fig:bvh} illustrates a BVH constructed for a 2D scene containing four distinct triangle primitives. For example, if a ray intersects primitive C, the traversal path to locate the intersection would be $B0 \rightarrow B2 \rightarrow B5 \rightarrow C$.
Notably, a ray may intersect multiple primitives, so its traversal path through the scene is not necessarily linear. In a 3D scene, the BVH is typically built by recursively dividing the space along the longest axis, creating two subspaces and forming corresponding AABBs.

Compared with R-trees and their variants, BVHs are preferred for ray casting because their hierarchical structure is designed specifically to optimize ray--primitive intersection queries. In a BVH, the bounding volumes of sibling nodes do not overlap, which ensures that each primitive resides in exactly one leaf node and prevents redundant traversal of the same geometry. This non-overlapping organization allows rays to move through a clean and efficient hierarchy without repeated tests. In addition, BVH traversal follows a simple and predictable pattern in which the algorithm descends the hierarchy, tests bounding volumes, and prunes subtrees that cannot be intersected. These properties make BVHs highly efficient for processing the millions of rays issued in modern graphics workloads and a natural fit for ray casting applications.

\begin{algorithm}[t]
\caption{\method{} Query for Facility $q$.}
\label{alg:rt-rknn}
\begin{algorithmic}[1]

\REQUIRE Facility $q$; user set $U$; facility set $F$; pruning threshold $k$.
\ENSURE RkNN users set of facility $q$: $\mathcal{R}_q$.

\vspace{0.5em}
\STATE \textbf{Scene Construction: Construct occluders.}
\STATE $\mathcal{B}_q \gets InfZone\_pruning(q,F,k)$.
\STATE $\mathcal{T}_q \gets \emptyset$
\FOR{each facility $f \in \mathcal{B}_q$}
    \STATE Build occluder $O_{f:q}$ based on Definition~\ref{def:triangle}.
    \STATE $\mathcal{T}_q \gets \mathcal{T}_q \cup O_{f:q}$
\ENDFOR
\STATE Build a BVH over $\mathcal{T}_q$.

\vspace{0.5em}
\STATE \textbf{Ray Casting: Count intersections along user rays.}
\STATE $\mathcal{R}_q \gets \emptyset$
\FOR{each user $u \in U$}
    \STATE Construct a ray $\mathbf{r}_u$ from $u$ based on Definition~\ref{eq:ray-u}.
    \STATE $c \gets 0$
    \STATE Traverse the BVH of $\mathcal{T}_q$.
    
    \FOR{each occluder $O$ intersected by $\mathbf{r}_u$}
        \STATE $c \gets c + 1$
        \IF{$c \ge k$}
            \STATE \textbf{break}   \COMMENT{Early exit: user is prunable}
        \ENDIF
    \ENDFOR
    \IF{$c < k$}
        \STATE $\mathcal{R}_q \gets \mathcal{R}_q \cup \{u\}$
    \ENDIF
\ENDFOR

\RETURN $\mathcal{R}_q$

\end{algorithmic}
\end{algorithm}

Algorithm~\ref{alg:rt-rknn} summarizes our graphics ray casting approach for processing \RkNN{} queries. The computation proceeds in two phases: \emph{scene construction}, where occluders associated with facilities are embedded into a 3D representation; and \emph{ray casting}, where rays representing users are issued to determine whether each user can be pruned.
During scene construction (Lines~1--8), for each facility \(f \in F \setminus \{q\}\), we compute the perpendicular bisector with \(q\) and encode the invalid side as occlusion triangles following Definition~\ref{def:triangle}. 
To avoid constructing unnecessary occluders and to eliminate the need for a candidate verification phase after filtering like SIX~\cite{six}, TPL~\cite{tpl} and SLICE~\cite{slice}, we apply InfZone-style pruning~\cite{InfZone}: once a facility's occluder is already fully covered by \(k\) previously constructed occluders, it is discarded because no ray can reach it. 
This pruning significantly reduces the number of remaining facilities, denoted \(m \ll |F|\), and enables full utilization of the GPU during ray casting rather than falling back to CPU-based candidate verification.
The construction of occluders requires \(O(m)\), while the pruning logic contributes \(O(m^2)\), dominated by intersection calculations between bisectors. 
A BVH is then built over the resulting scene \(\mathcal{T}_q\). 
Unlike facility R-trees used in prior work, which can be reused across different queries, the BVH structure must be rebuilt for each query because the occlusion geometry depends on the specific query facility.
While BVH construction could be amortized when the facility set is stable, similar strategies could equally benefit TPL and InfZone via cached bisector sets; for fairness, we apply no such amortization to any method.
As BVH construction typically requires \(O(m \log m)\)~\cite{bvh_c1,bvh_c2,bvh_c3}, the overall complexity of the scene construction phase becomes \(O(m^2 + m \log m)\).

After the scene is built, ray casting begins (Lines~9--24). Each user \(u\) emits a ray \(\mathbf{r}_u\) into the scene, and intersections are evaluated
accordingly. 
Early termination is applied: once \(k\) intersections are detected, the user is pruned (Line~18). Since a single intersection search under BVH traversal has an expected cost of \(O(\log m)\)~\cite{bvh_traversal}, processing a ray with at most \(k\) allowable intersections results in a total expected complexity of \(O(k \log m)\) for determining whether the user belongs to the \RkNN{} result.
It is worth noting that, unlike prior work that indexes users using an R-tree or its variants, our approach does not index users at all. 
While this increases the asymptotic cost of the ray casting phase to \(O(k|U|\log m)\), it enables full exploitation of parallelism, allowing all rays to be processed independently, especially by GPUs associated with RT cores and yielding substantial performance benefits in practice.

\subsection{Implementing \method{} using RT cores}
\label{subsec:optix}

To fully leverage hardware acceleration, we implement the proposed formulation algorithm using RT cores.

RT cores were originally designed to accelerate ray tracing in computer graphics, a more complex form of ray casting that incorporates effects such as reflection, refraction, and shadows. 
Before it was introduced, ray tracing was difficult to scale on GPUs because different rays typically produce different numbers of ray--primitive intersections, resulting in irregular memory access and divergent control flow that conflict with the SIMT execution model. RT cores address these challenges by providing hardware-accelerated ray--triangle and ray--AABB intersection tests during BVH traversal and by scheduling similar intersection tasks together to reduce divergence and improve throughput. After formulating \RkNN{} queries as ray casting, our algorithm can be implemented on RT cores with relatively modest programming effort using their corresponding APIs.

There are several RT APIs, including Nvidia OptiX~\cite{optix}, AMD HIP-RT~\cite{HIP-RT}, DirectX Raytracing~\cite{DXR} and Vulkan~\cite{vulkan}.
All of these APIs share a similar programming model. In this paper, we choose Nvidia OptiX for our implementation due to its cutting-edge performance and strong compatibility with CUDA. It is worth noting that while the implementation in this paper targets Nvidia GPUs’ RT cores, porting the algorithm to AMD GPUs, whose ray tracing hardware is functionally similar, is straightforward.
To use OptiX, developers follow a predefined workflow:
\begin{enumerate}
\item Create a GPU context.
\item Build a BVH.
\item Create a program pipeline.
\item Build the Shader Binding Table (SBT).
\item Launch rays.
\end{enumerate}

\begin{algorithm}[t]
\caption{\method{} Implementation using OptiX.}\label{alg:ray-casting}
\begin{algorithmic}[1]
\REQUIRE The constructed scene $\mathcal{T}_q$, scene height $m$. 
\ENSURE User verification result $isRkNN$.
\vspace{0.5em}
\STATE \textbf{procedure} \texttt{Ray Generation}
    \STATE \hspace{0.5cm} $t_{\min}, t_{\max} \gets \{0, m+1\}$
    \STATE \hspace{0.5cm} $u \gets \texttt{optixGetLaunchIndex()}$
    \STATE \hspace{0.5cm} $\mathbf{o} \gets (u_x,u_y,m+1)$
    \STATE \hspace{0.5cm} $\mathbf{d} \gets (0,0,-1)$
    \STATE \hspace{0.5cm} $c \gets 0$
    \STATE \hspace{0.5cm} $isRkNN \gets $ TRUE
    \STATE \hspace{0.5cm} \texttt{optixTrace}$(\mathcal{T}_q, \mathbf{o}, \mathbf{d}, t_{\min}, t_{\max}, \texttt{payloads}(c,isRkNN))$
\STATE \textbf{end procedure} 
\vspace{0.5em}
\STATE \textbf{procedure} \texttt{Any-hit}
    \STATE \hspace{0.5cm}  $c \gets c + 1$
    \STATE \hspace{0.5cm} \textbf{if}{$ c < k $} \textbf{then}
    \STATE \hspace{0.5cm}\hspace{0.5cm}  \texttt{optixIgnoreIntersection()}
    \STATE \hspace{0.5cm} \textbf{else}
    \STATE \hspace{0.5cm}\hspace{0.5cm}  $isRkNN \gets$ FALSE
    \STATE \hspace{0.5cm}\hspace{0.5cm}  \texttt{optixTerminateRay()}
    \STATE \hspace{0.5cm} \textbf{end if}
\STATE \textbf{end procedure} 
\end{algorithmic}
\end{algorithm}

A context is first created to associate OptiX with a specific CUDA device.
In step (2), a BVH corresponding to the scene is constructed.
The program pipeline in step (3) consists of several programmable stages:
\begin{itemize}[leftmargin=3.5mm]
    \item \textit{Ray generation}: The entry point to the pipeline, defines the ray by origin point $\mathbf{o}$ and direction $\mathbf{d}$ based on equation~(\ref{eq:ray}).
    \item \textit{Intersection}: Implements a ray-primitive intersection test, invoked during BVH traversal. Programmers do not need to implement it when using the default hardware intersection test.
    \item \textit{Any-hit}: Called when a new, potentially closest intersection of a ray is found. Not necessary when only the closest intersection is counted.
    \item \textit{Closest-hit}: Called when the ray finds the closest intersection.
    \item \textit{Miss}: Called after confirming a ray does not hit any primitive in the scene.
\end{itemize}
OptiX uses a single ray programming model~\cite{optix}, where each ray executes this pipeline independently.
The SBT constructed in step (4) links primitives with their associated programs and serves as the binding layer between rays and scene behavior.
Finally, rays are launched and will explore the scene based on the program pipeline defined before.


As shown in Algorithm~\ref{alg:ray-casting}, our OptiX implementation of ray casting stage is straightforward because the formulated \RkNN{} computation aligns naturally with ray tracing. 
Only the \texttt{Ray Generation} and \texttt{Any-hit} programs are required. 
We use \texttt{Any-hit} instead of \texttt{Closest-hit} because we are interested in counting up to \(k\) intersections per ray rather than identifying the nearest one. 
The \texttt{Miss} program can be omitted, as no additional processing is required when a ray misses all occluders. The \texttt{Intersection} program is left
empty to leverage the default hardware-accelerated triangle intersection routine provided by the RT cores.
OptiX will trace the ray just generated by calling \texttt{optixTrace()}.

User's coordinates are embedded in the rays which can be obtained by \texttt{optixGetLaunchIndex()}. The height of the ray origins is set to $m+1$ to ensure it is higher than all occluders in the scene. During traversal, once \(c\) reaches \(k\), the ray is terminated early using \texttt{optixTerminateRay()}. Otherwise, intermediate intersections invoke \texttt{optixIgnoreIntersection()}, allowing traversal to continue. 
Under OptiX’s single ray programming model, the \texttt{Ray Generation} program launches one ray per user and stores both the intersection count \(c\) and the user verification result \(\textit{isRkNN}\) via OptiX \texttt{payloads} variables.

\section{Evaluation}\label{sec:performance}
This section outlines the experimental setup and reports performance with detailed analysis.

\subsection{Experimental Setup}
\textbf{System}:
We compare our \method{} approach against three baseline algorithms introduced in \autoref{sec:related_work}: TPL, InfZone, and SLICE. In the figures, InfZone is abbreviated as \textit{INF}, and our \method{} method is abbreviated as \textit{RT} for clarity.
All algorithms were implemented from scratch, using shared common routines for comparable operations. 
To emphasize query-time efficiency, all baselines use R*-trees rather than standard R-trees, providing a stronger query-time foundation for the competing methods.
The only exception is the BVH, which is constructed exclusively for our \method{} method.
Implementations are written in C++11 and compiled using NVCC~12.4 with -O2 flag, and OptiX~7.7 enabled where applicable. 
Experiments were conducted on a workstation equipped with an AMD EPYC~7453 28-core CPU, 512~GB of system memory, and an Nvidia RTX~A6000 GPU with 48~GB of GPU memory, running Ubuntu 22.04.5~LTS.

\begin{figure}[t]
    \centering
    \begin{subfigure}[b]{0.31\columnwidth}
        \centering
        \includegraphics[width=\textwidth]{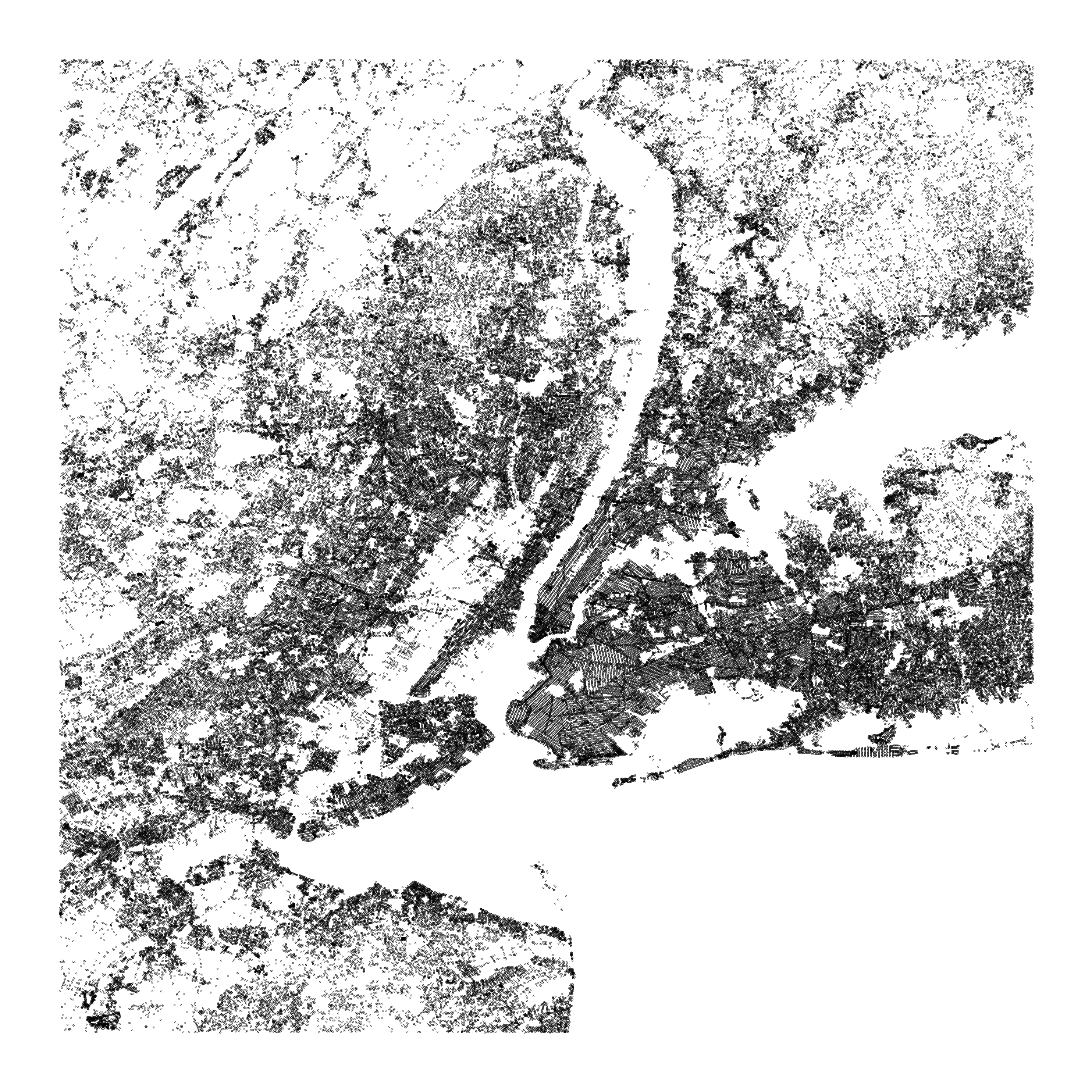}
        \caption{NY}
        \label{fig:sub1}
    \end{subfigure}
    \hfill
    \begin{subfigure}[b]{0.31\columnwidth}
        \centering
        \includegraphics[width=\textwidth]{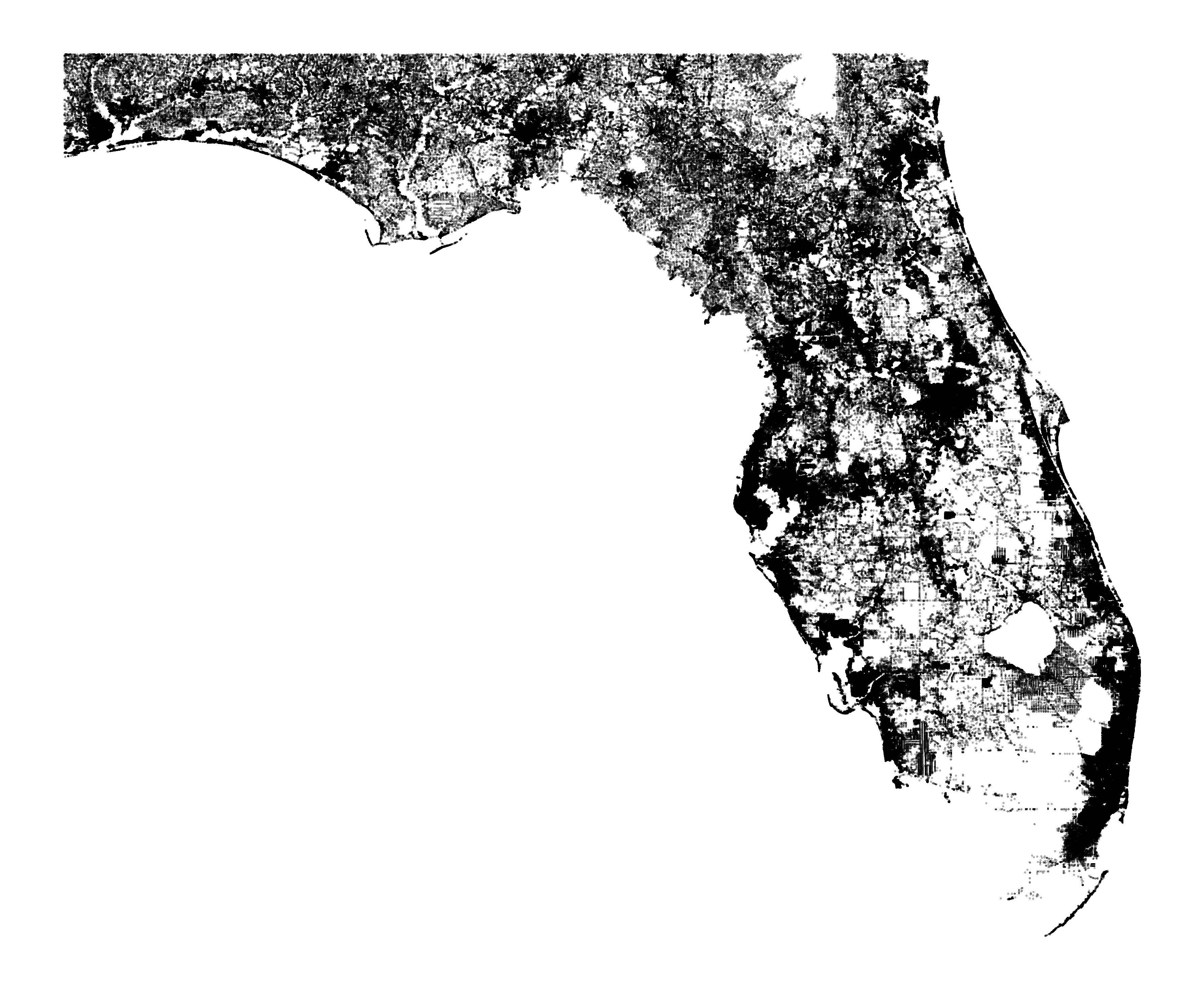}
        \caption{FLA}
        \label{fig:sub2}
    \end{subfigure}
    \hfill
    \begin{subfigure}[b]{0.31\columnwidth}
        \centering
        \includegraphics[width=\textwidth]{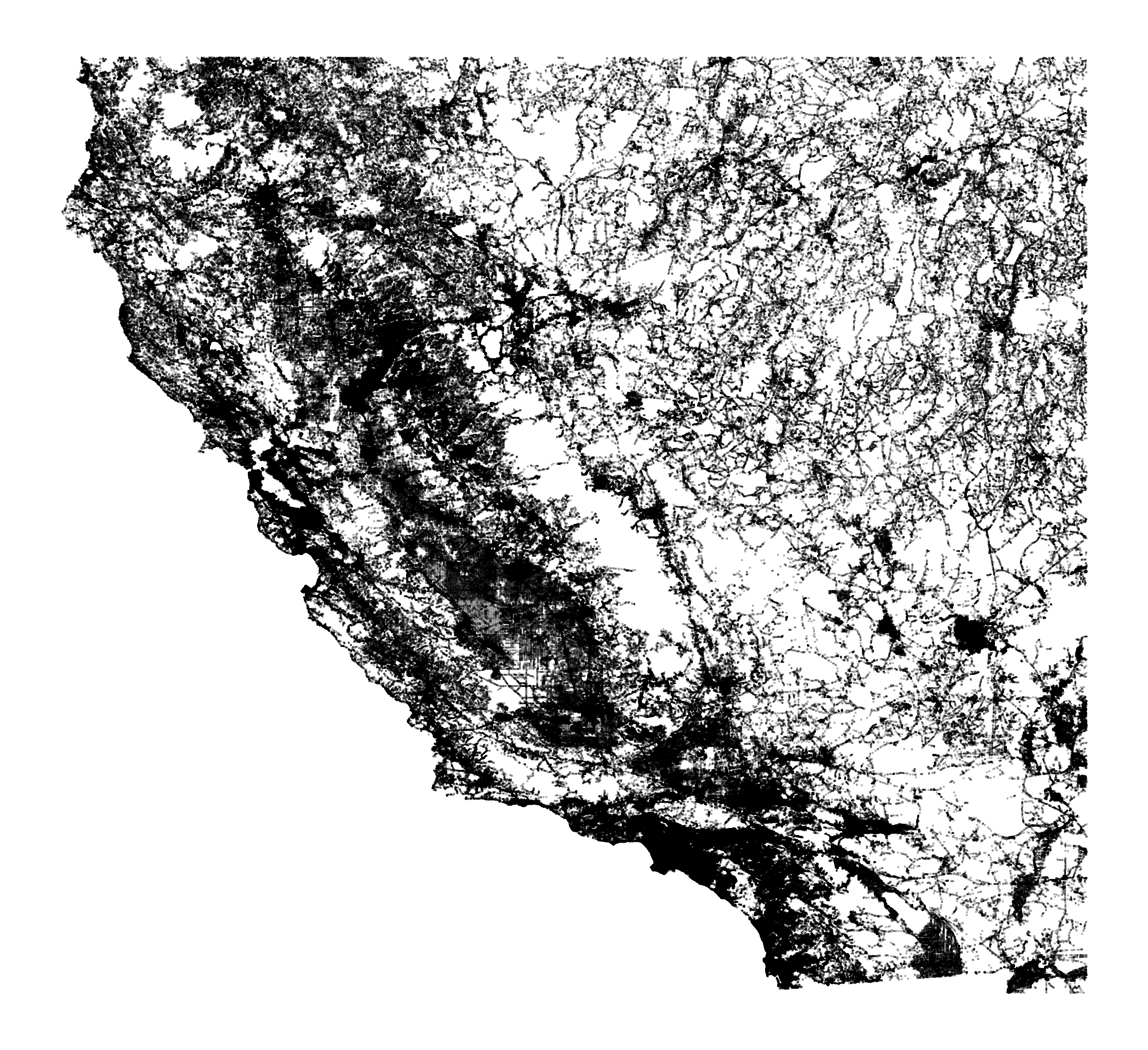}
        \caption{CAL}
        \label{fig:sub3}
    \end{subfigure}
    
    \caption{Visualization of datasets.}
    \label{fig:datasets}
\end{figure}

\begin{table}[t]
\centering
\caption{Six real-world road-network datasets.}
\label{tab:dataset}
\begin{tabular}{l l r}
\toprule
\textbf{Symbol} & \textbf{Description} & \textbf{Number of Points} \\
\midrule
USA  & Full USA                  & 23,947,347 \\
CTR  & Central USA               & 14,081,816 \\
E    & Eastern USA               & 3,598,623 \\
CAL  & California and Nevada     & 1,890,815 \\
FLA  & Florida                   & 1,070,376 \\
NY   & New York City             & 264,346 \\
\bottomrule
\end{tabular}
\end{table}

\textbf{Datasets}:
We use six real-world road-network datasets covering New York City (NY), Florida (FLA), California and Nevada (CAL), Eastern USA (E), Central USA (CTR), and the full USA (USA). These datasets contain between 264,346 and 23,947,347 points and were obtained from the DIMACS repository~\cite{dataset}. The spatial distributions of representative datasets are visualized in \autoref{fig:datasets}, and detailed statistics are provided in \autoref{tab:dataset}.

\textbf{Evaluation Configurations}:
Following prior work, we vary $k$ from 1 to 25 (default $k=10$), with additional values up to 200 to examine large-$k$ scalability. We use two facility settings: a \emph{default} setting of 1,000 randomly selected facilities (aligning with existing baselines) and a \emph{sparse} setting of 100 facilities (representing sparse facility distributions).
For both settings, except for points selected as facilities, all remaining points are used as users.
Results are averaged over 100 random queries for the sparse setting and 1,000 queries for settings with 1,000 or more facilities, ensuring statistical robustness and consistency. All algorithms (including \method{}) follow a two-stage execution model whose stage names and internal logic differ across methods; we adopt the convention of~\cite{tplpp}, calling the first stage \emph{filtering} and the second \emph{verification}, to enable direct runtime-breakdown comparisons across algorithms.

\begin{table}[t]
\centering
\caption{Amortized user indexing cost in dataset USA}
\label{tab:amorized}
\begin{tabular}{lrrrrrr}
\toprule
\textbf{Algorithms} &\textbf{Operation} &  \textbf{Runtime (s)} \\
\midrule
All baselines  & R*-tree construction  & 147.241 \\

\method{} (ours)  & Plain GPU transfer  & 0.010251 \\
\bottomrule
\end{tabular}
\end{table}

\subsection{Amortized Operation Costs}
In all baseline algorithms, R*-trees built over facilities and users are amortized across queries and thus excluded from the performance scaling analysis. \method{} likewise benefits from amortization: although it constructs a facility R*-tree to support InfZone-style pruning during scene construction, it does not require a user R*-tree. Instead, the user set is uploaded once to GPU memory and reused throughout the workload. The same applies to OptiX context and program pipeline creation, which is amortizable and trivial. Although amortizable costs do not contribute to per-query scalability, they remain relevant when comparing overall preprocessing burden across algorithms, especially in workloads with short query sequences. As shown in \autoref{tab:amorized}, \method{} incurs significantly less indexing overhead than the baselines while retaining the advantages of amortized preprocessing, making it comparatively lightweight in preprocessing cost.

\subsection{Impact of Varying k Settings}

\autoref{fig:k-sf} and \autoref{fig:k-nf} analyze the impact of \(k\) on runtime under the sparse facility and default facility settings,
respectively. 
Although only a subset of datasets is shown, all datasets exhibit the same overall trend.
Under the sparse facility setting, \method{} consistently outperforms all baseline algorithms except when \(k = 1\), where all methods complete in
approximately 3~ms. 
Under the default facility density, where baseline algorithms benefit from stronger spatial pruning opportunities, \method{} still outperforms TPL and InfZone across all tested values of \(k\). 
Even when compared with SLICE, which is known to be more resilient to increases in \(k\), \method{} begins to outperform it once \(k\) reaches 10--20.
For the extremely large $k$ setting, we only compare \method{} with SLICE on our largest dataset USA.
\autoref{fig:largek} illustrates that for extremely large $k$ scenarios, our \method{} outperforms the state-of-the-art algorithm SLICE in all cases.
Remarkably, our \method{} achieves up to 16.4$\times$ speedup over SLICE, 34.6$\times$ over InfZone, and as much as 54.7$\times$ over TPL when $|F| = 100$ and $k = 25$ on the CAL dataset.
This result reflects the issue discussed in \autoref{sec:related_work}: as \(k\) increases, pruning-based methods lose pruning effectiveness, resulting in larger candidate sets and substantially more verification work. 
In contrast, for \method{}, increasing \(k\) simply allows a ray to continue traversal rather than terminating early. 
The algorithm, therefore, degrades gracefully, requiring only additional intersection counting while still benefiting from hardware acceleration. 
Consequently, the impact of larger \(k\) values on runtime remains modest compared with traditional approaches.

\begin{figure}[t]
\centering

\begin{minipage}[t]{0.48\textwidth}
    \centering
    \begin{subfigure}[b]{0.48\textwidth}
        \centering
        \includegraphics[width=\textwidth]{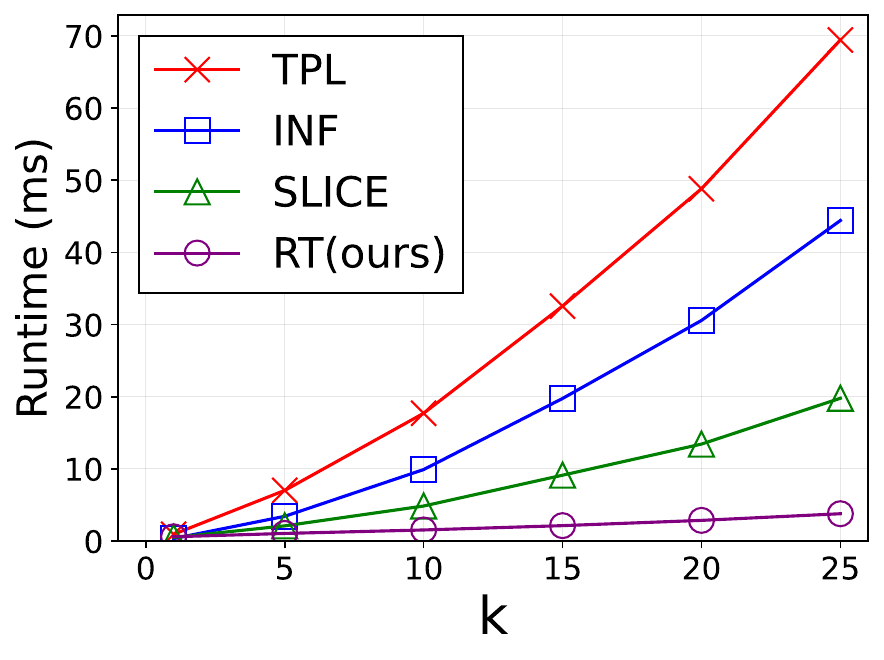}
        \caption{NY (setting $|F|=100$)}
        \label{fig:k_ny_f100}
    \end{subfigure}
    \hfill
    \begin{subfigure}[b]{0.48\textwidth}
        \centering
        \includegraphics[width=\textwidth]{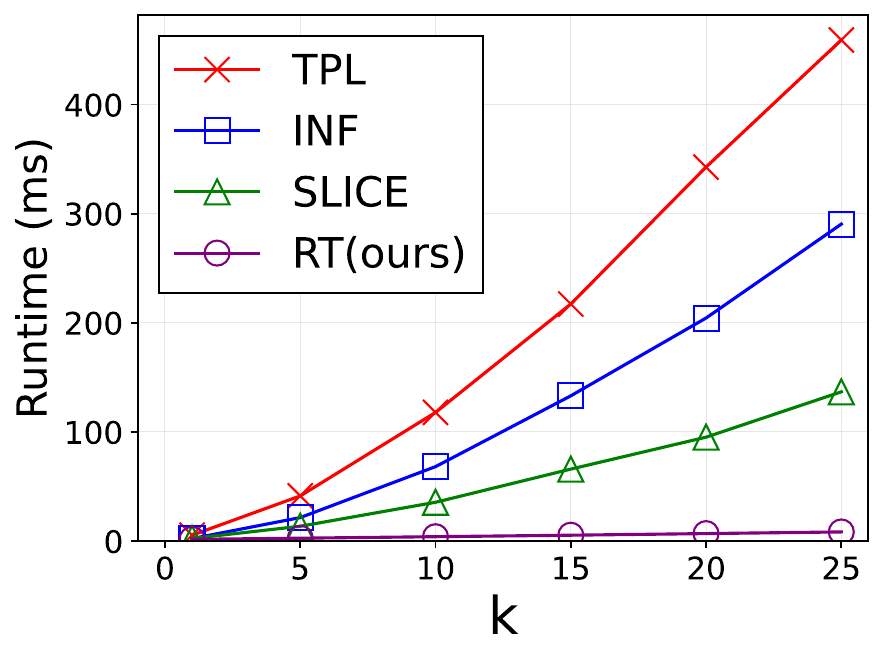}
        \caption{CAL (setting $|F|=100$)}
        \label{fig:k_cal_f100}
    \end{subfigure}
    \caption{Impact of varying $k$ on runtime in sparse facility setting.}
    \label{fig:k-sf}
\end{minipage}
\hfill
\begin{minipage}[t]{0.48\textwidth}
    \centering
    \begin{subfigure}[b]{0.48\textwidth}
        \centering
        \includegraphics[width=\textwidth]{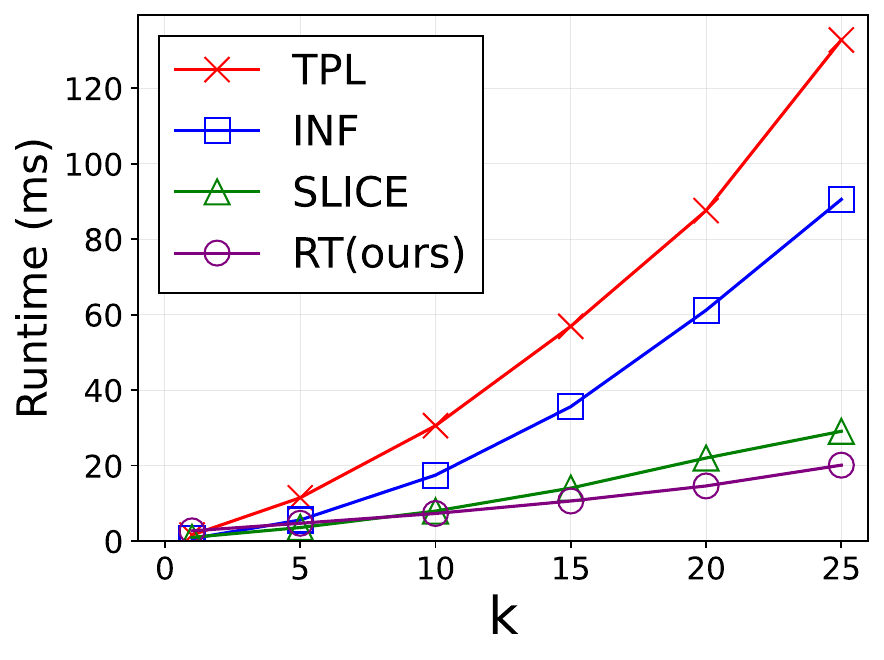}
        \caption{E (setting $|F|=1000$)}
        \label{fig:k_E_f1000}
    \end{subfigure}
    \hfill
    \begin{subfigure}[b]{0.48\textwidth}
        \centering
        \includegraphics[width=\textwidth]{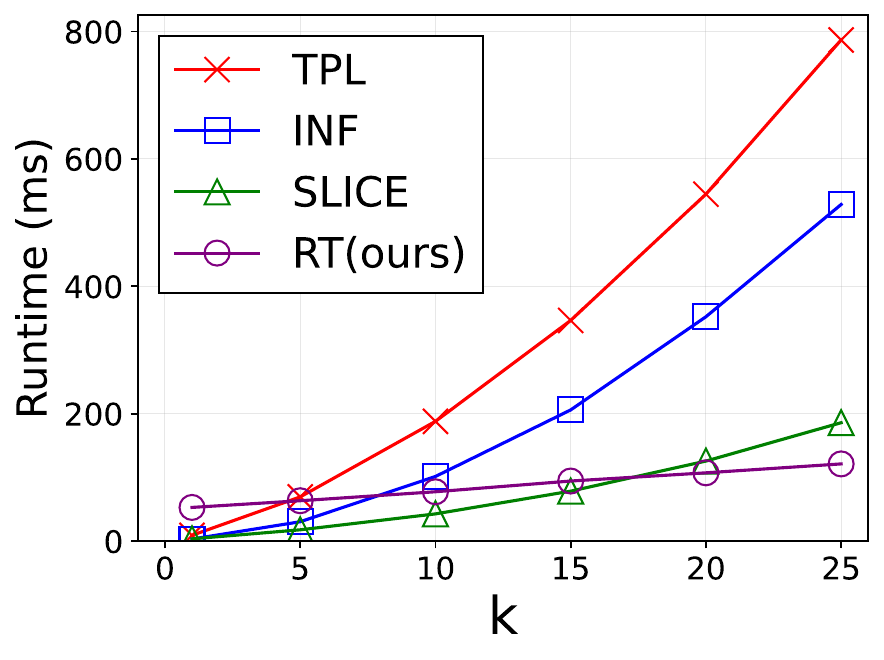}
        \caption{USA (setting $|F|=1000$)}
        \label{fig:k_USA_f1000}
    \end{subfigure}
    \caption{Impact of varying $k$ settings on runtime using the default facility setting.}
    \label{fig:k-nf}
\end{minipage}

\end{figure}
\begin{figure}[t] 
  \centering \includegraphics[width=1\linewidth]{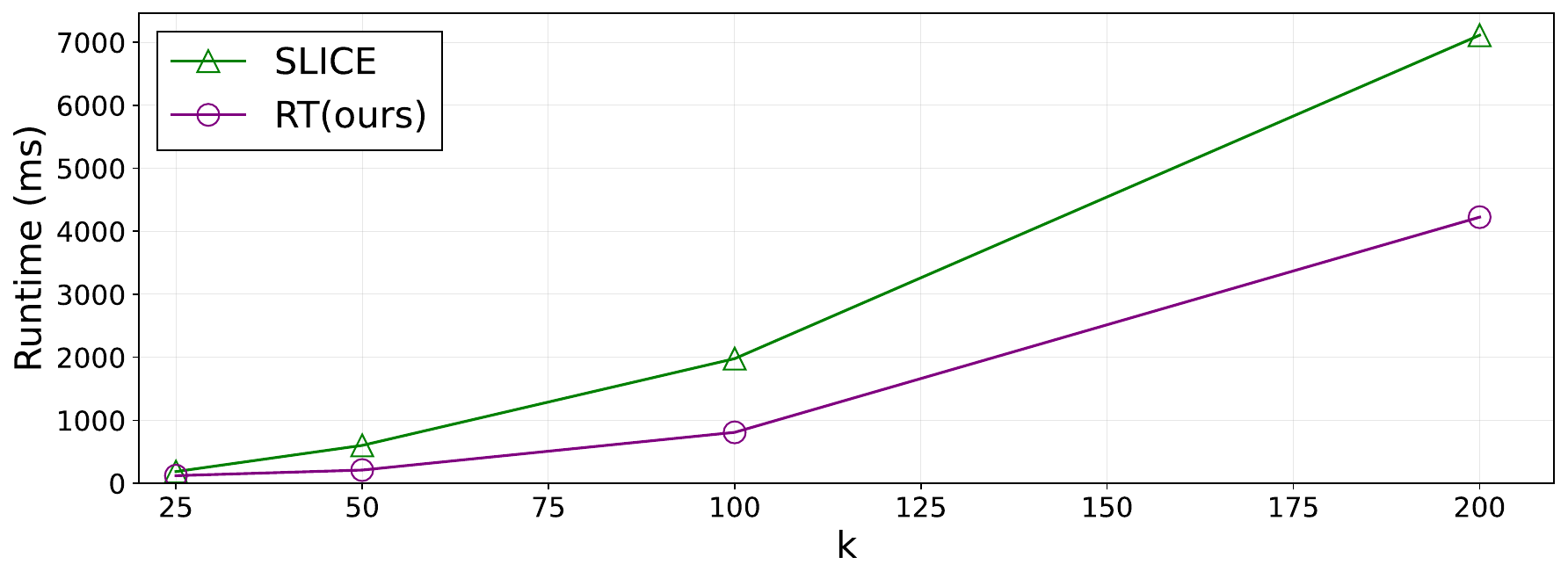} 
  \caption{Impact of large $k$ on runtime in default facility setting ($|F|=1000$) of USA dataset.} 
  \Description{large k.} 
  \label{fig:largek} 
\end{figure}

\subsection{Impact of Varying Data Sizes}

\begin{figure}[t]
    \centering
    \begin{subfigure}[b]{1\columnwidth}
        \centering
        \includegraphics[width=\textwidth]{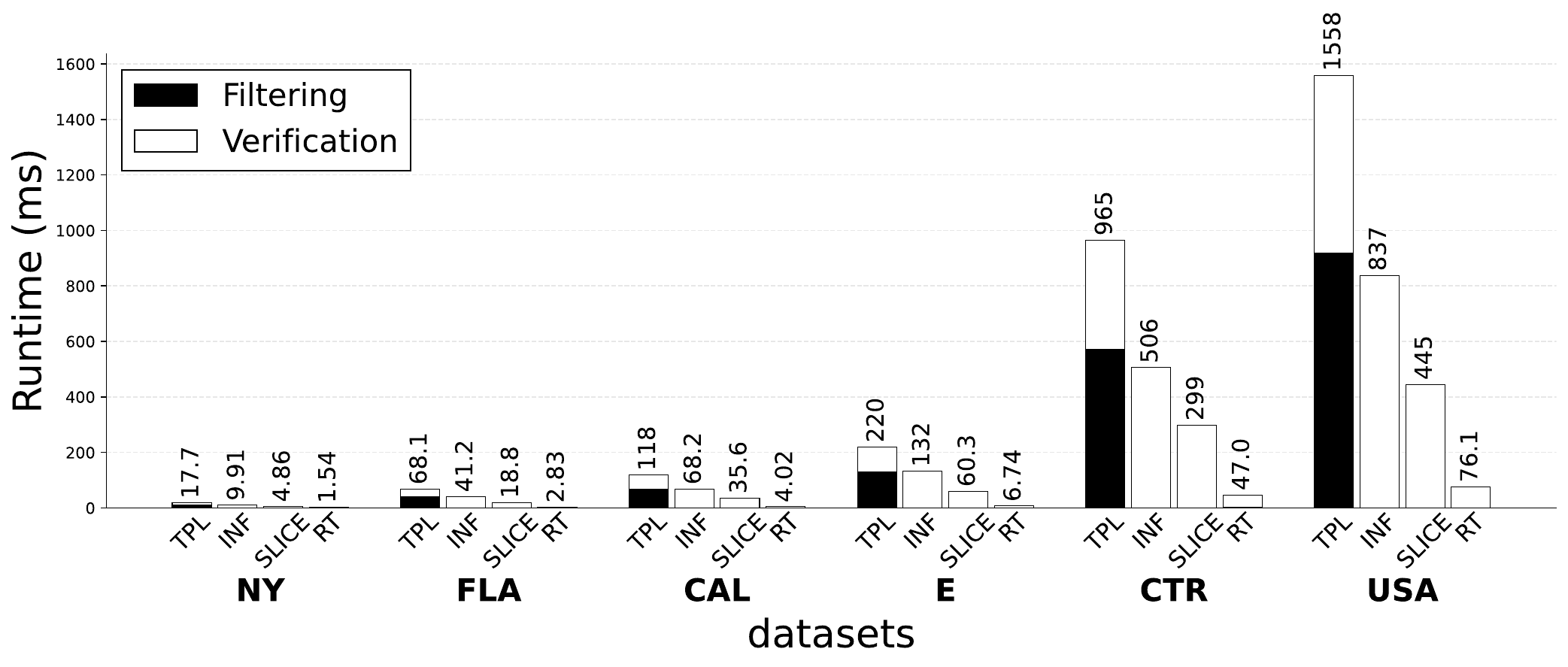}
        \caption{Sparse facility setting.}
        \label{fig:datasize_f100}
    \end{subfigure}
    \hfill
    \begin{subfigure}[b]{1\columnwidth}
        \centering
        \includegraphics[width=\textwidth]{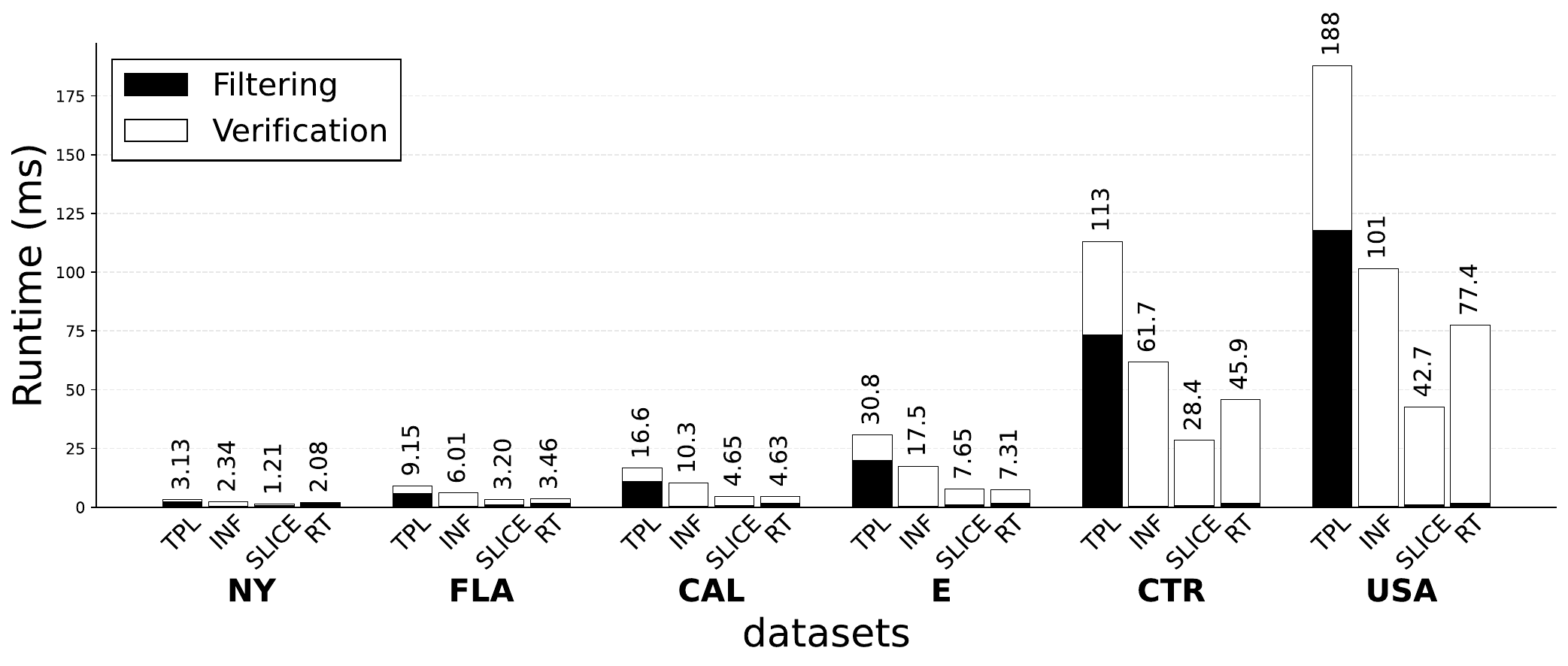}
        \caption{Default facility setting.}
        \label{fig:datasize_f1000}
    \end{subfigure}
    
    \caption{Impact of data size on runtime in sparse ($|F|=100$) and default ($|F|=1000$) facility setting.}
    \label{fig:datasize}
\end{figure}

We study the impact of data size on each algorithm using all 6 datasets, which span a wide range of scales, with the largest dataset (USA) being about 100 times larger than the smallest (NY), under both the sparse facility and default facility settings.
In a sparse facility setting, as shown in \autoref{fig:datasize_f100}, our \method{} outperforms all baselines across all datasets. 
Its runtime increases only slowly as data size grows, demonstrating the effectiveness of our algorithm with massively parallel ray--primitive intersection processing on GPU RT cores. In contrast, the performance of baseline algorithms degrades substantially once the dataset exceeds approximately 10 million points (datasets E and USA), reflecting their limited scalability under large input sizes.
\autoref{fig:datasize_f1000} shows that under the default facility setting, where spatial pruning remains highly effective, \method{} does not outperform SLICE. This is primarily due to the additional overhead introduced by data transfer between main memory and GPU memory.

\subsection{Impact of Varying Facility Cardinality}

\begin{figure}[t]
\centering

\begin{minipage}[t]{0.48\textwidth}
    \centering
    \begin{subfigure}[b]{0.48\columnwidth}
        \centering
        \includegraphics[width=\textwidth]{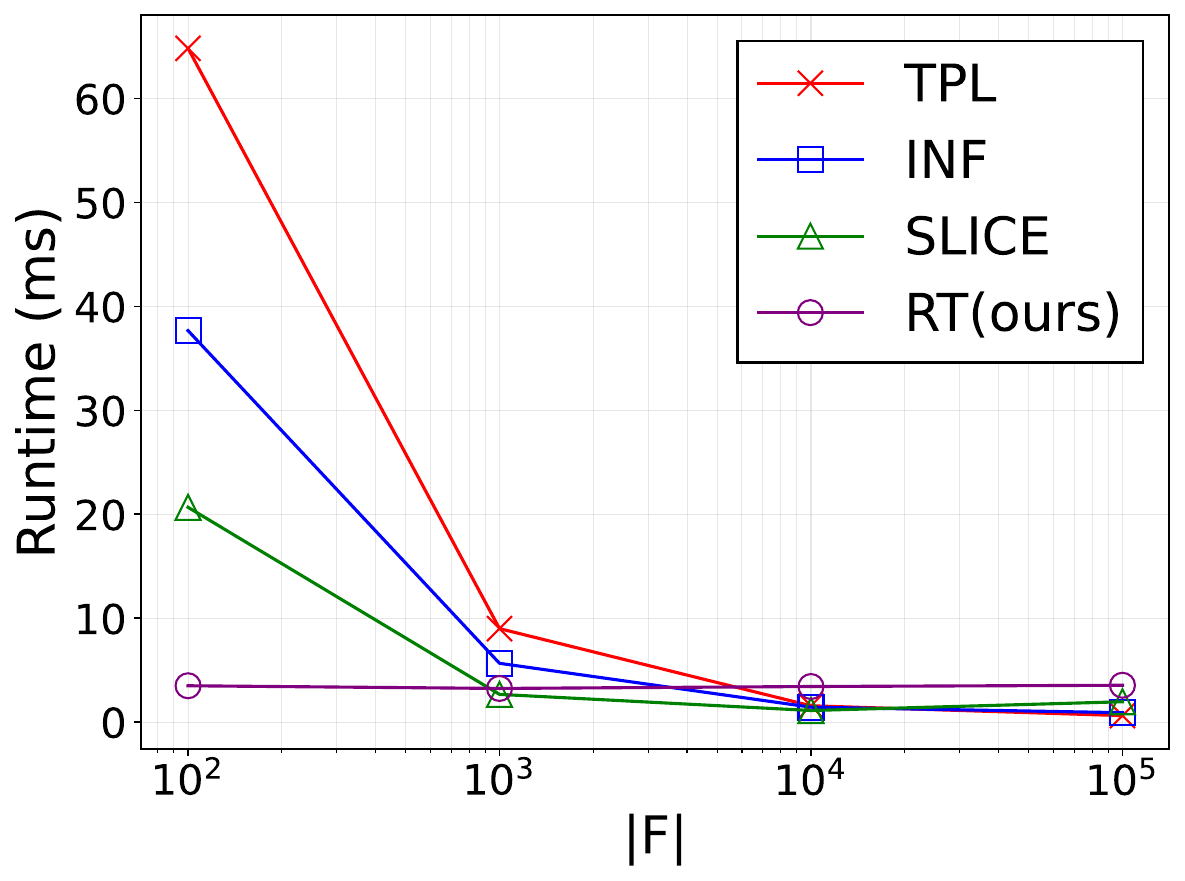}
        \caption{CAL ($|U|=10^6$)}
        \label{fig:f_cal_line}
    \end{subfigure}
    \hfill
    \begin{subfigure}[b]{0.48\columnwidth}
        \centering
        \includegraphics[width=\textwidth]{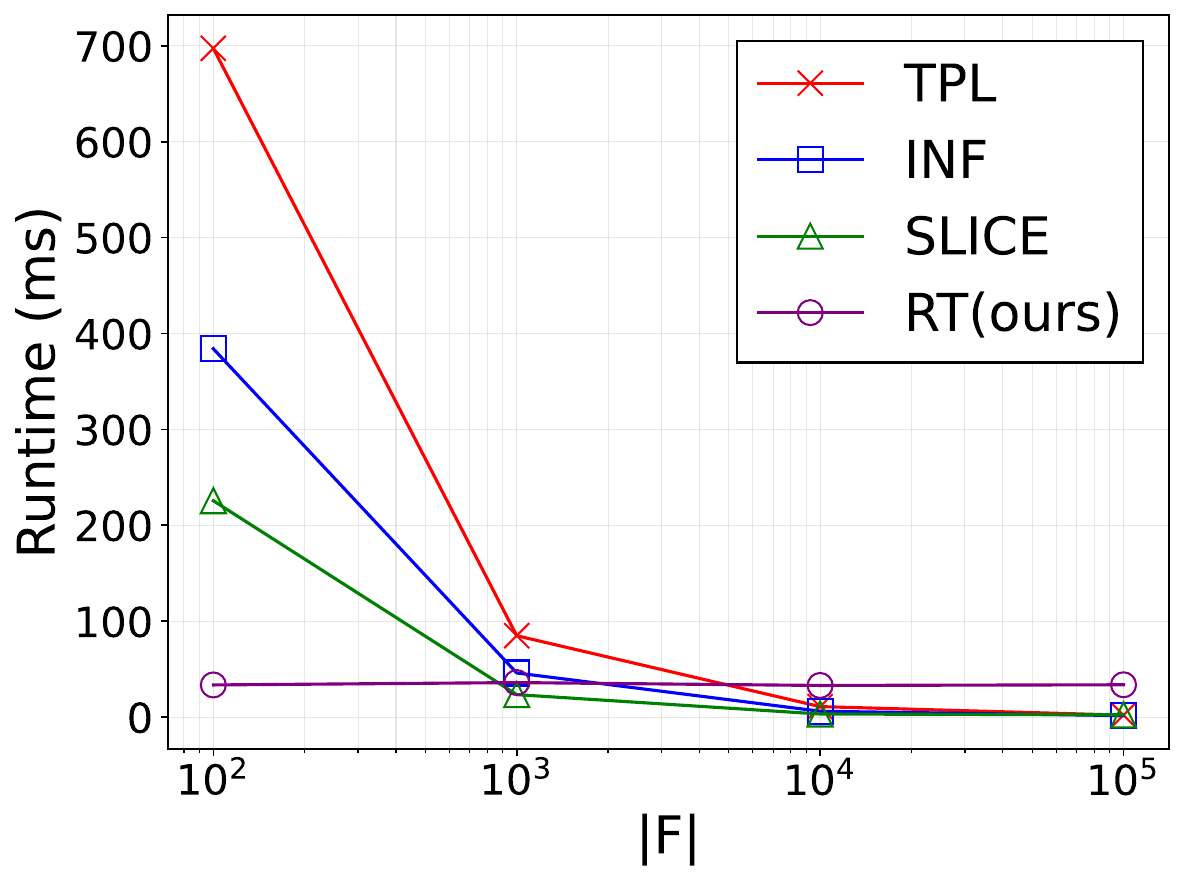}
        \caption{USA ($|U|=10^7$)}
        \label{fig:f_usa_line}
    \end{subfigure}
    
    \caption{Total runtime under varying cardinality for datasets CAL ($|U|=10^6$) and USA ($|U|=10^7$).}
    \label{fig:ficility}
\end{minipage}
\hfill
\begin{minipage}[t]{0.48\textwidth}
    \centering
    \begin{subfigure}[b]{1\columnwidth}
        \centering
        \includegraphics[width=\textwidth]{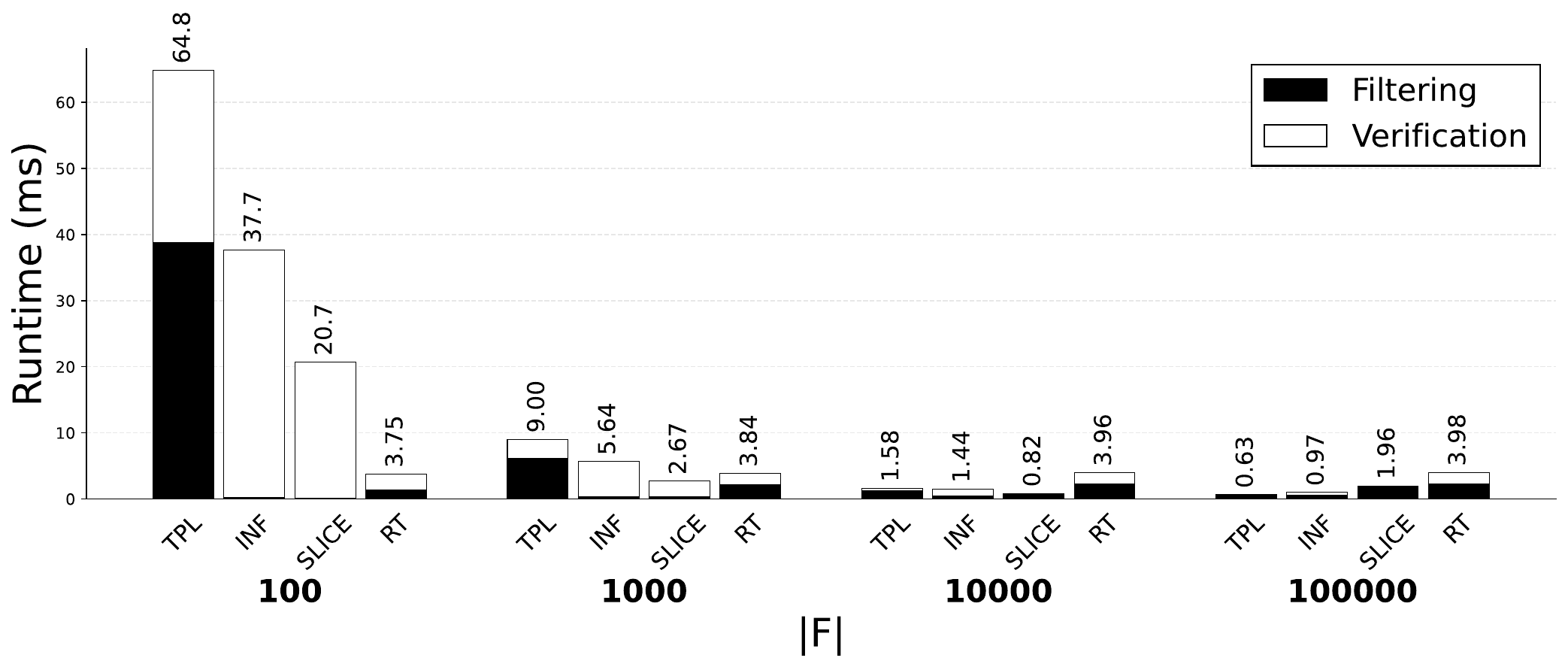}
        \caption{CAL ($|U|=10^6$)}
        \label{fig:f_cal}
    \end{subfigure}
    \hfill
    \begin{subfigure}[b]{1\columnwidth}
        \centering
        \includegraphics[width=\textwidth]{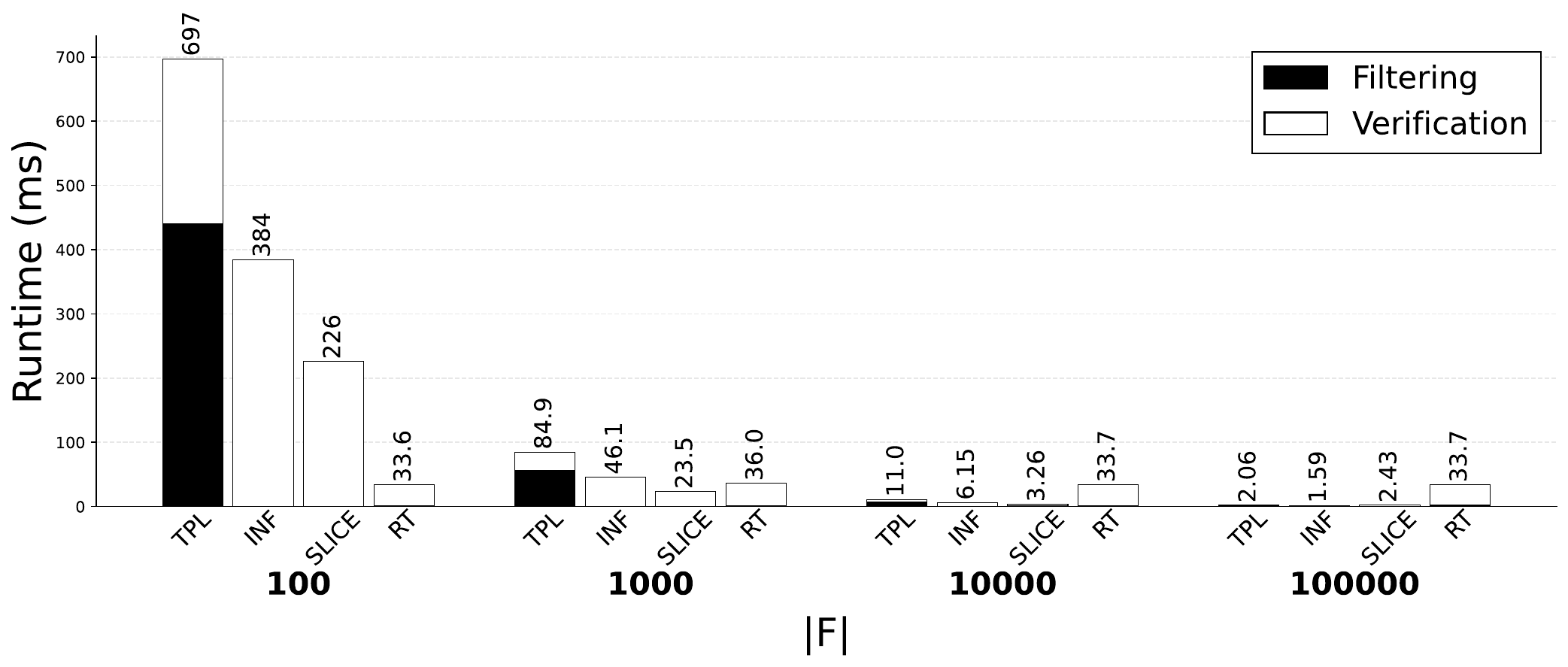}
        \caption{USA ($|U|=10^7$)}
        \label{fig:f_usa}
    \end{subfigure}
    
    \caption{Breakdown of filtering and verification time under varying facility cardinality for datasets CAL ($|U|=10^6$) and USA ($|U|=10^7$).}
    \label{fig:ficility_bar}
\end{minipage}

\end{figure}

In this section, we fix the number of users to 1\,million for CAL and 10\,million for USA, and vary the facility cardinality \(|F|\) to study its impact on performance.
From \autoref{fig:ficility}, we observe that \method{} maintains nearly constant performance across all datasets, regardless of the number of facilities. 
In contrast, the baseline algorithms improve significantly as \(|F|\) grows, benefiting from effective spatial pruning. 
From \autoref{fig:ficility_bar}, we observe that in both cases, the verification phase of the baseline algorithms becomes significantly faster as the facility cardinality increases, while \method{} maintains a stable verification cost.
The flat performance curve of \method{} occurs because all computation that can benefit from GPU parallelism is already highly optimized, while the remaining overhead involved by the GPU RT cores pipeline, is inherently difficult to parallelize and therefore does not improve with larger facility sets. 
As a result, \method{} remains largely unaffected by increases in \(|F|\), whereas pruning-based methods gain efficiency from having more facilities.
For \textit{monochromatic \RkNN{} queries}, where the query point and candidate points originate from the same set, spatial relationships tend to be more structured, providing a more favorable environment for spatial pruning as shown here.
Under these conditions, \method{} does not surpass the SLICE algorithm, which benefits greatly from strong pruning effectiveness.

\subsection{Impact of Varying User Cardinality}

\begin{figure}[t]
\centering

\begin{minipage}[t]{0.48\textwidth}
    \centering
    \begin{subfigure}[b]{0.48\columnwidth}
        \centering
        \includegraphics[width=\textwidth]{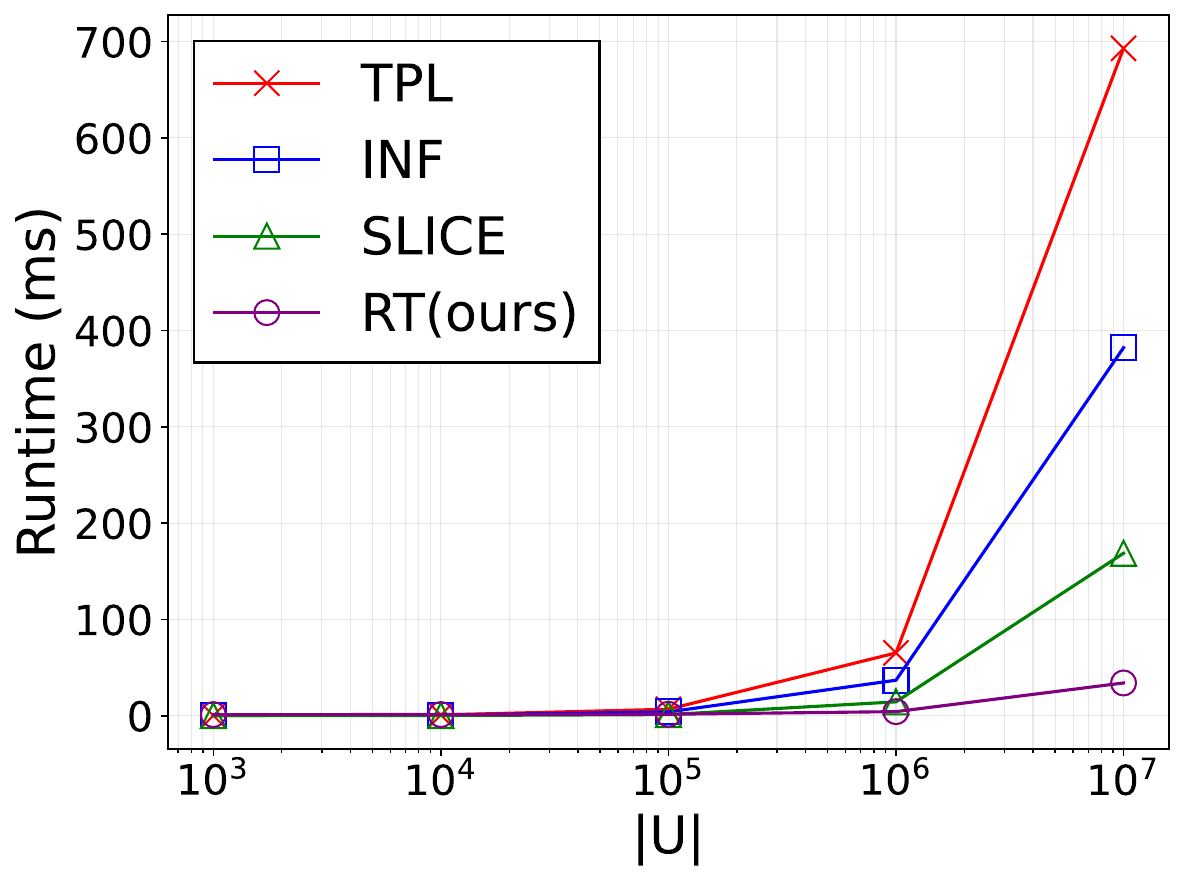}
        \caption{Sparse facility setting.}
        \label{fig:u_f100_line}
    \end{subfigure}
    \hfill
    \begin{subfigure}[b]{0.48\columnwidth}
        \centering
        \includegraphics[width=\textwidth]{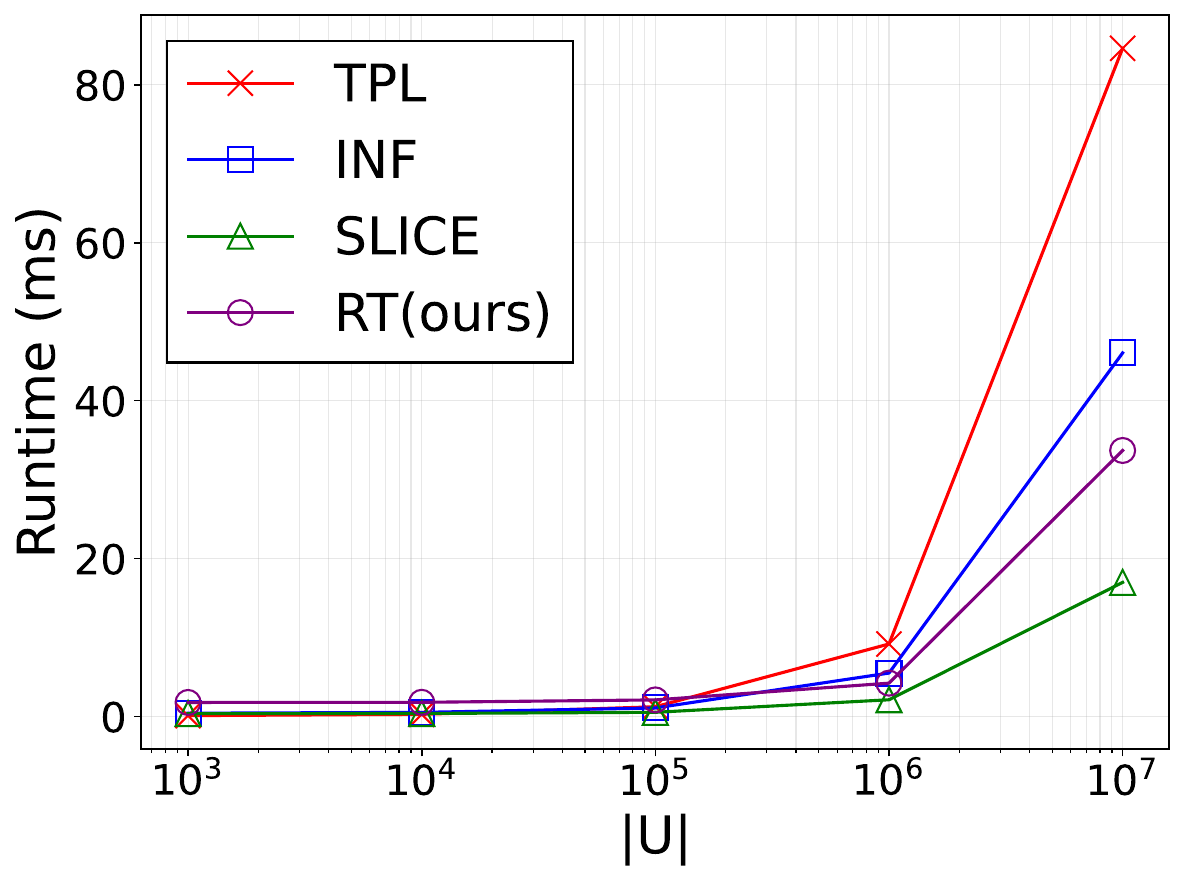}
        \caption{Default facility setting.}
        \label{fig:u_f1000_line}
    \end{subfigure}
    
    \caption{Total runtime under varying facility cardinality for sparse and default facility setting of dataset USA.}
    \label{fig:user}
\end{minipage}
\hfill
\begin{minipage}[t]{0.48\textwidth}
    \centering
    \begin{subfigure}[b]{1\columnwidth}
        \centering
        \includegraphics[width=\textwidth]{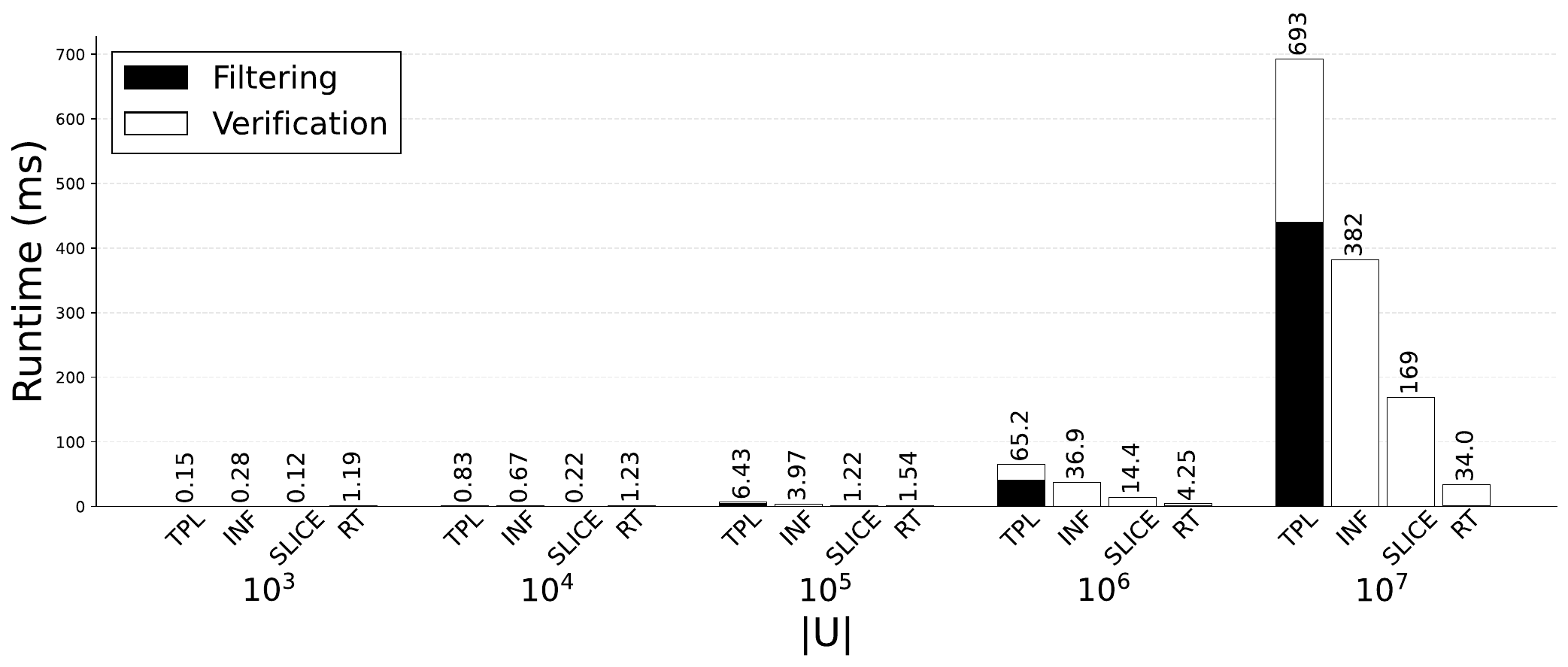}
        \caption{Sparse facility setting.}
        \label{fig:u_usa_f100}
    \end{subfigure}
    \hfill
    \begin{subfigure}[b]{1\columnwidth}
        \centering
        \includegraphics[width=\textwidth]{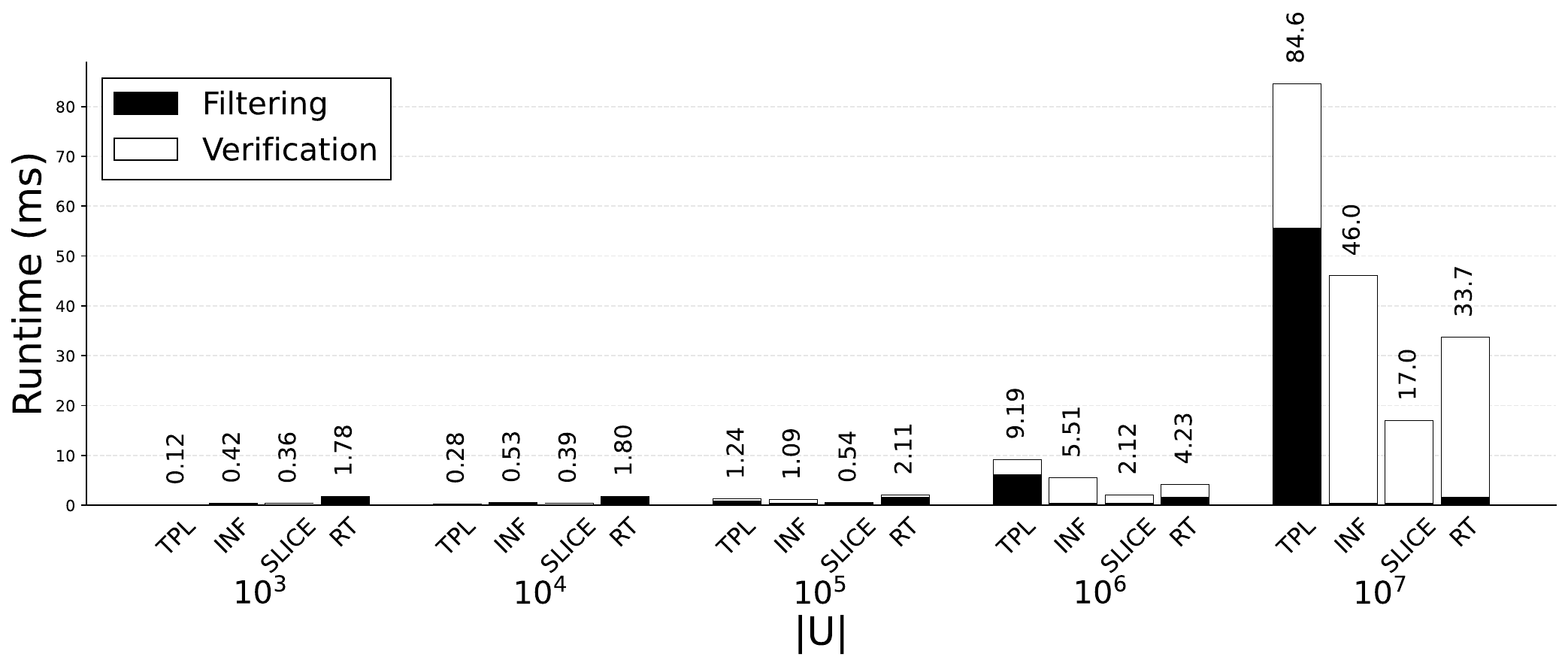}
        \caption{Default facility setting.}
        \label{fig:u_usa_f1000}
    \end{subfigure}
    
    \caption{Breakdown of filtering and verification time under varying user cardinality for the USA dataset in sparse and default facility settings.}
    \label{fig:user_bar}
\end{minipage}

\end{figure}

In this section, we fix the number of facilities using both the sparse and default facility settings for the USA dataset, and vary the user cardinality \(|U|\) to examine its impact on performance.
All algorithms perform well when the number of users is small.
However, performance degrades for every method as the user population becomes large. 
In the sparse facility setting, as shown in \autoref{fig:u_f100_line}, \method{} maintains strong performance even with very large user populations of up to \(10^7\) points, outperforming the baseline algorithms by a substantial margin. 
In the default facility setting, as illustrated in \autoref{fig:u_f1000_line}, \method{} still does not surpass SLICE.
As shown in \autoref{fig:user_bar}, in the sparse setting, \method{} achieves the best performance due to its highly efficient verification phase, which benefits from massive GPU parallelization. However, this advantage does not carry over to the default facility setting, where SLICE becomes faster but \method{} experiences performance degradation, especially when $|U|$ reaches $10^{7}$ due to the additional overhead introduced by GPU data transfer. 
Nonetheless, although \method{} does not outperform SLICE in every scenario, the performance gap in the default facility setting is relatively small compared to the substantial advantages \method{} achieves in the sparse facility setting.

\begin{figure}[t]
    \centering
    \begin{subfigure}[b]{0.48\columnwidth}
        \centering
        \includegraphics[width=\textwidth]{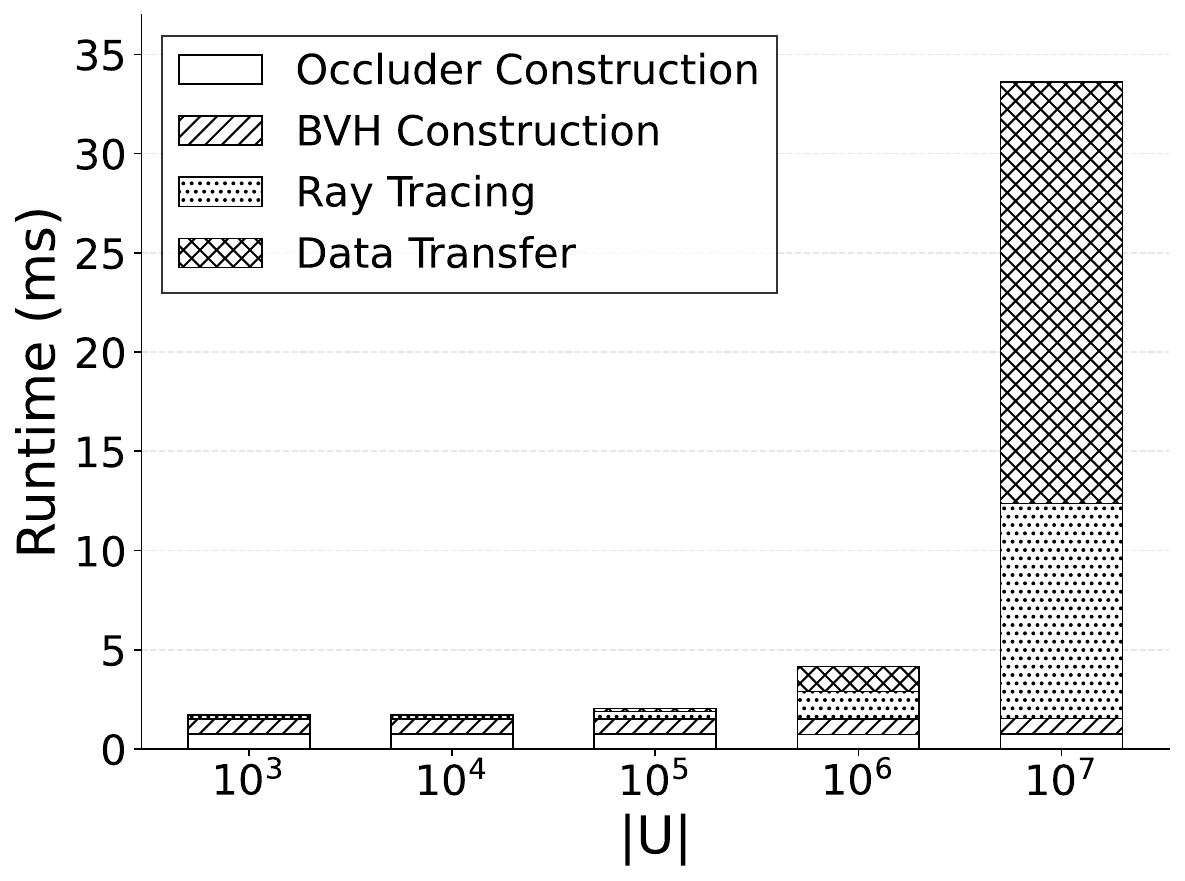}
        \caption{USA ($|F|=1000, k=10$)}
        \label{fig:u-breakdown}
    \end{subfigure}
    \hfill
    \begin{subfigure}[b]{0.48\columnwidth}
        \centering
        \includegraphics[width=\textwidth]{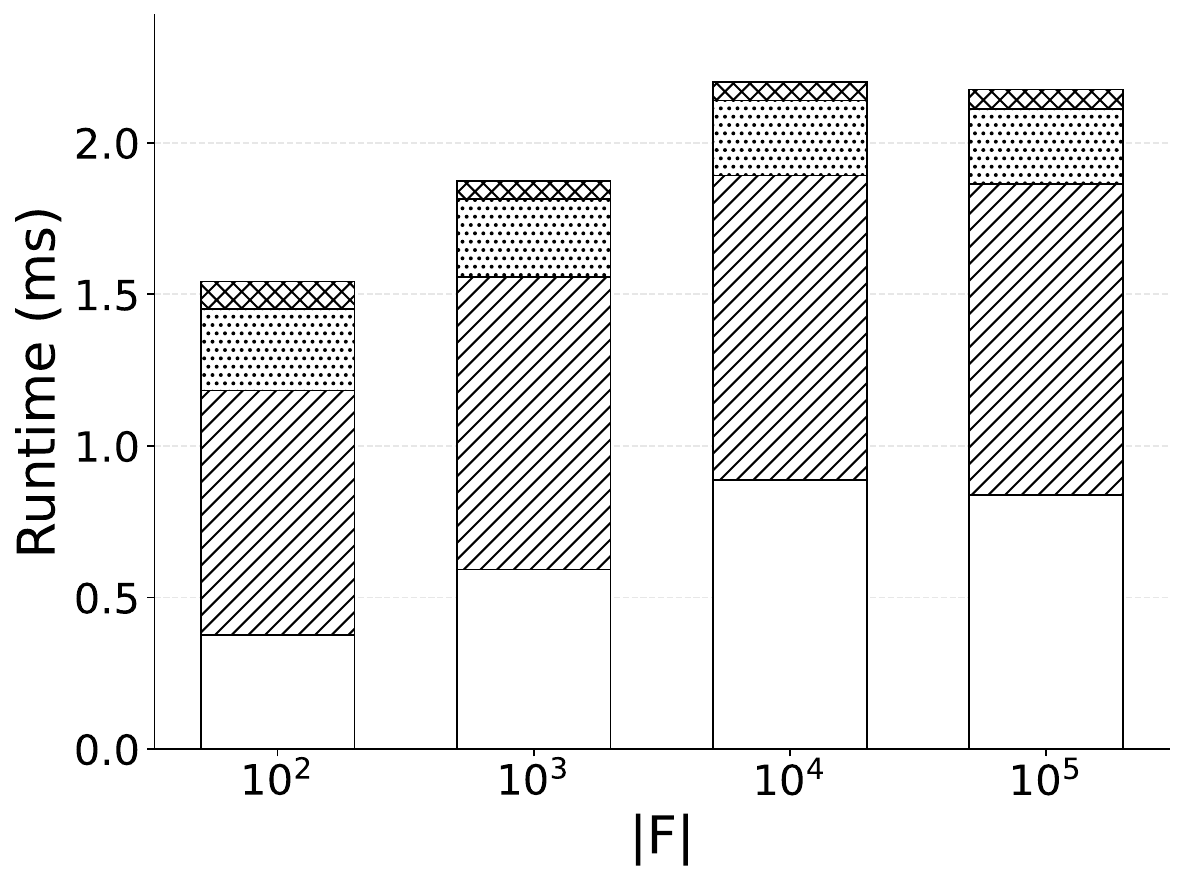}
        \caption{CAL ($|U|=10^5, k=10$)}
        \label{fig:f-breakdown}
    \end{subfigure}
    
    \caption{Runtime breakdown analysis under configurations with fixed $|U|$ and fixed $|F|$.}
    \label{fig:breakdown}
\end{figure}

\subsection{Runtime Breakdown Analysis}
In this section, we provide a breakdown analysis of \method{} to examine its advantages and limitations. 
According to the OptiX workflow, the \emph{Scene Construction} stage can be divided into two components: 
occluder construction and BVH construction. 
Although the size of the Shader Binding Table (SBT) also depends on the scene configuration, we do not include the SBT
construction time in our analysis, because \method{} only requires hit or miss information and does not rely on per-primitive attributes, making SBT construction trivial. 
In practice, SBT setup takes approximately \(0.062\) ms even for the largest scene where \(|F| = 10^7\), which is negligible compared to the other costs.
Then, \emph{Ray Casting} stage is divided into ray tracing and data transfer, which copies the ray tracing result back to the main memory from the GPU.

We fix \(|F| = 10^3\) and vary \(|U|\) from \(10^3\) to \(10^7\) to evaluate how user cardinality affects the performance of \method{} on the USA dataset. 
Since the facility set is fixed, the constructed scene remains identical across all experiments. 
As shown in \autoref{fig:u-breakdown}, RT cores sustain high ray tracing efficiency until the user population reaches approximately one million. 
However, when \(|U|\) increases to ten million, ray tracing time and data transfer time begin to dominate. 
This is expected because \(10^7\) rays exceed the number that the GPU can process concurrently, causing some rays to wait while others are executed.
It is important to note that data transfer constitutes pure overhead, influenced solely by user cardinality. When \(|U| = 10^7\), data transfer accounts for more than half of the total runtime. 
This overhead becomes particularly significant in dense facility settings, where baseline algorithms achieve efficient \RkNN{} computation with small \(k\). 
This explains why \method{} cannot outperform SLICE on very large datasets under default facility densities.
We then fix \(|U| = 10^5\) and vary \(|F|\) from \(10^2\) to \(10^5\) to evaluate how facility cardinality affects the performance of \method{} on the CAL dataset. The user cardinality is restricted to \(10^5\) to ensure that the \emph{Ray Casting} stage does not dominate the overall runtime. 
As \(|F|\) increases, the occluder construction time initially grows and reaches its peak around \(|F| = 10^4\). 
Beyond this point, the construction time stabilizes or even slightly decreases because InfZone-style pruning removes most unnecessary occluders. 
This also explains why BVH construction time and ray tracing time remain largely unaffected by changes in
\(|F|\).

\subsection{Impact of Different Occluder Counts}
In this section, we investigate how the number of occluders affects the performance of \method{}.
As introduced earlier, we apply InfZone-style pruning to avoid constructing unnecessary occluders. 
However, in practice, we observe that when the facility set is small, constructing all occluders without pruning can actually be faster due to simpler control flow and better data locality. 
Motivated by this observation, we conduct supplementary experiments comparing three occluder construction strategies:
{
\begin{itemize}[leftmargin=4.5mm]
    \item \textit{InfZone-style pruning}:  
    The original pruning strategy introduced by InfZone, which aggressively removes as many unnecessary occluders as possible, with $O(m^2)$ complexity.
    \item \textit{Conservative pruning}:  
    InfZone-style pruning is applied only for the first few occluders (e.g., the first 20). Beyond that, a lightweight conservative test based on Equation~(\ref{eq:cheapinf}) is used. Although not optimal, this approach still prunes most unnecessary occluders while substantially reducing pruning overhead.
    \item \textit{Non-pruning}:  
    As the name suggests, this strategy constructs all occluders without any pruning.
\end{itemize}
}

\begin{figure}[t] 
  \centering \includegraphics[width=1\linewidth]{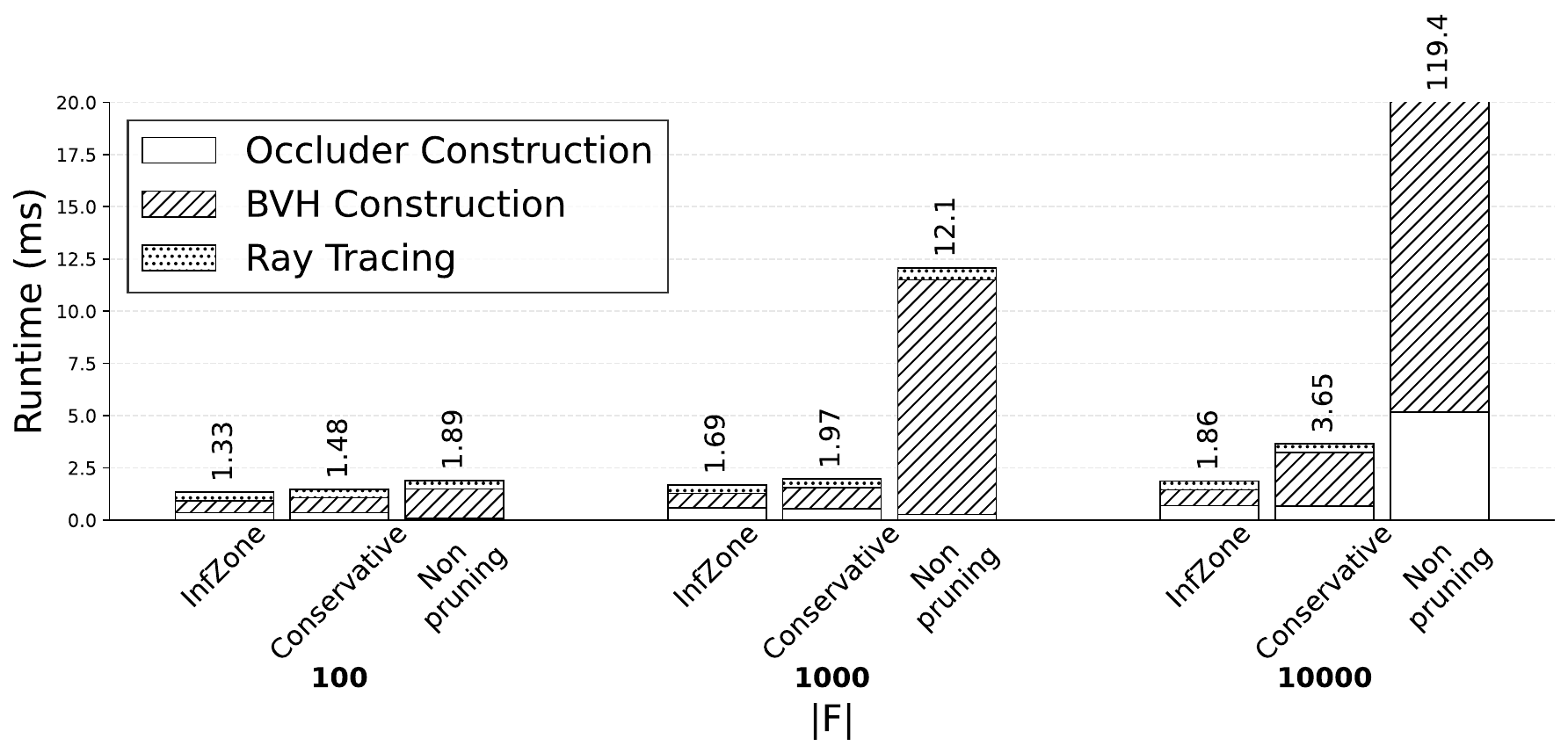} 
  \caption{Impact of occluder counts on runtime under the default New York facility configuration.} 
  \Description{large k.} 
  \label{fig:occluders} 
\end{figure}

\begin{table}[t]
\centering
\caption{Average occluder counts of different occluder construction strategies in different facility cardinality.}
\begin{tabular}{lccc}
\toprule
 & $|F|=10^2$ & $|F|=10^3$ & $|F|=10^4$ \\
\midrule
InfZone-style & 37.27 & 46.34 & 48.44 \\
Conservative & 49.54 & 83.04 & 205.34 \\
Non-pruning & 99 & 999 & 9999 \\
\bottomrule
\end{tabular}
\label{tab:occluder_count}
\end{table}

The average occluder count produced by each pruning strategy under different facility cardinalities on the NY dataset is reported in \autoref{tab:occluder_count}, and the corresponding performance results are shown in \autoref{fig:occluders}.
Since our goal is to examine performance variations attributable specifically to the number of occluders, \autoref{fig:occluders} includes only the components directly affected: \textit{Occluder Construction}, \textit{BVH Construction}, and \textit{Ray Tracing}.

We observe that \textit{Ray Tracing} is not strongly affected by changes in the occluder count, except for the non-pruning strategy when $|F|=9999$, which strongly complicates the scene for a ray to traverse (did not show in figure for clarity of other settings). This is because the NY dataset has a relatively small user cardinality, allowing rays to be processed with high parallel efficiency. 
The \textit{No-pruning} strategy constructs occluders very quickly. However, the resulting large number of occluders
significantly slows down subsequent stages, ultimately making it the slowest among all strategies even in a sparse facility setting.
For \textit{Conservative pruning}, although it removes approximately 98\% of unnecessary occluders when \(|F| = 10^4\), the remaining occluders still constitute roughly four times the number produced by \textit{InfZone-style pruning}. 
While \textit{Conservative pruning} is slightly faster than \textit{InfZone-style pruning} in the \textit{Occluder Construction} stage, this advantage is offset during BVH construction, resulting in inferior overall performance.

Consequently, \textit{InfZone-style pruning} emerges as the most effective strategy for \textit{Scene Construction}, providing the shortest total runtime and lower GPU memory usage.

\subsection{Performance Evaluation w/o RT Cores}
Since no publicly available GPU-based \RkNN{} algorithms exist, we implement a GPU baseline, termed \textit{InfZone-GPU}, which directly offloads InfZone’s verification phase to the GPU without using RT cores. 
We choose InfZone because it is the only baseline that produces no false positives, thereby eliminating the need for additional user verification and allowing most computation to remain on the GPU.
We also include the original InfZone algorithm in our comparison to demonstrate the computational advantage provided by GPU execution.

\begin{figure}[t] 
  \centering \includegraphics[width=1\linewidth]{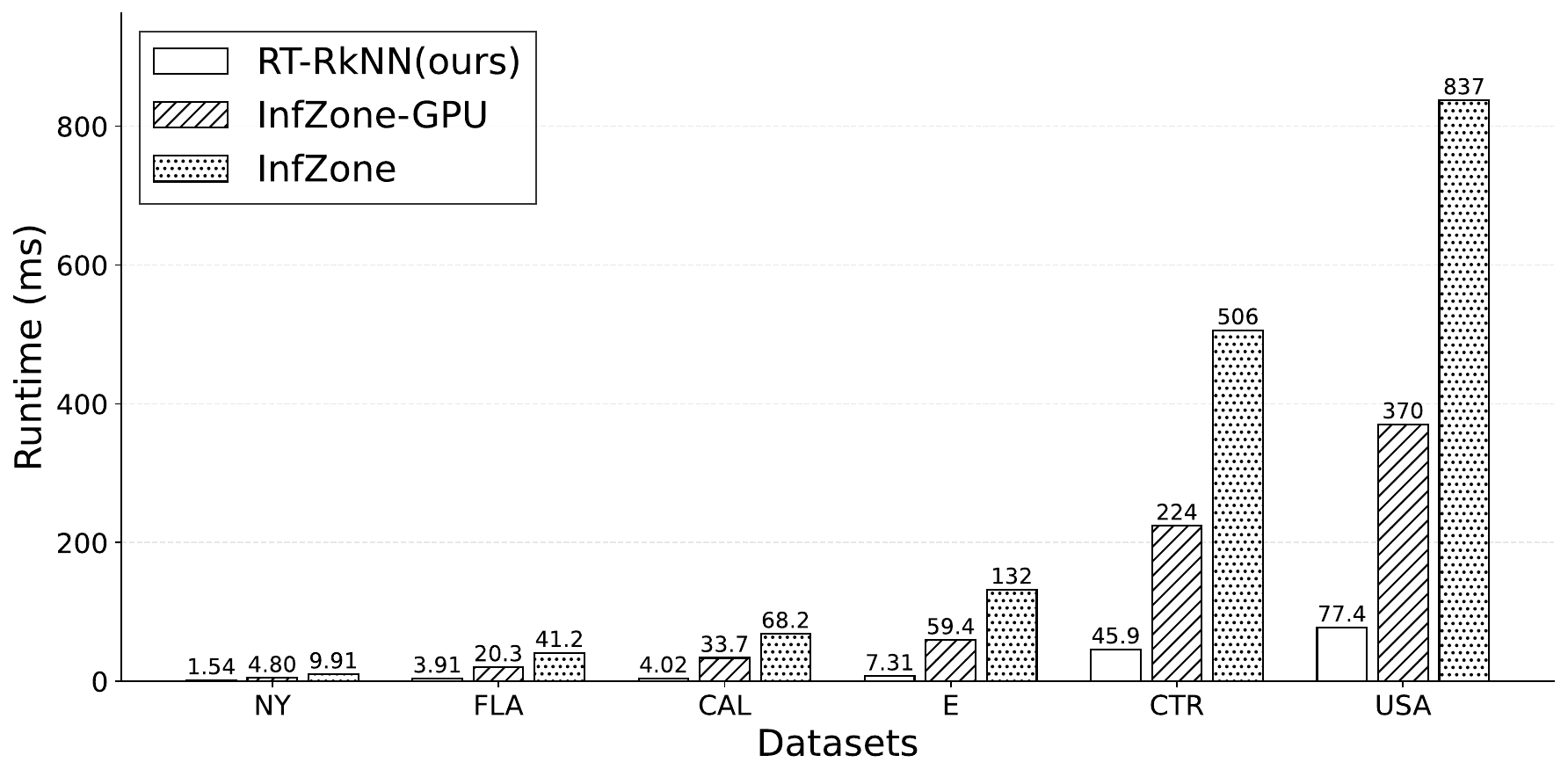} 
  \caption{Performance comparison between \method{}, GPU-based InfZone without RT cores, and CPU-based InfZone in the sparse facility setting of the NY dataset.} 
  \Description{large k.} 
  \label{fig:gpu} 
\end{figure}

On one hand, from \autoref{fig:gpu}, we observe that InfZone-GPU outperforms the original CPU-based InfZone across all datasets, demonstrating the GPU's substantial computing power even when the verification phase contains significant control divergence and irregular memory accesses that are poorly aligned with the SIMT execution model. 

On the other hand, \method{} consistently outperforms InfZone-GPU across all datasets. As data size increases, the performance gap widens: InfZone-GPU degrades rapidly due to the growing control divergence across users during verification, whereas \method{} maintains stable scaling behavior. This is because \method{} leverages dedicated hardware intersection testing in RT cores, which eliminates the need for branch-heavy verification logic and drastically reduces control divergence. These results demonstrate that our ray casting formulation, which enables algorithm fit for RT cores hardware acceleration, is essential for achieving high performance in GPU-based \RkNN{} computation.

\section{Conclusions and Future Work}\label{sec:conclusion}
In this paper, we presented \method{}, which formulates \RkNN{} as graphics ray casting by modeling users as rays and facilities as occluders, establishing a formal equivalence between spatial pruning and ray--primitive intersection. Building on this, we design the first algorithm directly exploiting GPU RT cores for hardware-accelerated intersection testing and BVH traversal. Experiments on real-world datasets with up to more than 23 million points show that \method{} significantly outperforms state-of-the-art algorithms, especially where traditional pruning becomes ineffective: sparse facility distributions, large user populations, and large $k$. \method{} also eliminates user indexing, cutting preprocessing overhead while enabling massive parallelism across all user rays.

For future work, we plan to explore four directions: (1) adapting the approach for dynamic and continuous \RkNN{} queries where facilities or users change over time; (2) investigating batched query processing to amortize scene construction costs across multiple query facilities; (3) exploring hybrid strategies that dynamically select between \method{} and traditional pruning based on data characteristics; and (4) implementing and evaluate this formulation on other RT hardware.

\begin{acks}
This work was supported in part by JSPS KAKENHI Grant Number JP24K20782.
\end{acks}


\balance
\bibliographystyle{ACM-Reference-Format}
\bibliography{sample}

\end{document}